%% file: succinct.tex
\documentclass[prodmode,acmtkdd]{acmsmall}

\usepackage{booktabs}
\usepackage{amsmath}
\usepackage{amssymb}
\usepackage{url}
\usepackage{comment}
\usepackage{rotating}
\usepackage{caption}
\usepackage{natbib}
\usepackage{nicefrac}

\usepackage[normalem]{ulem}

\usepackage[ruled, linesnumbered]{algorithm2e}

\SetAlFnt{\small}
\SetAlCapFnt{\small}
\SetAlCapNameFnt{\small}
\SetAlCapHSkip{0pt}
\IncMargin{-\parindent}

\acmYear{2012}

\input{defines}

\begin{document}

\newcommand{\ourtitle}{Summarizing Data Succinctly with the Most Informative Itemsets}

\markboth{M.\ Mampaey et al.}{\ourtitle}

\title{\ourtitle}
\author{MICHAEL MAMPAEY
\affil{University of Antwerp}
JILLES VREEKEN
\affil{University of Antwerp}
NIKOLAJ TATTI
\affil{University of Antwerp}
}
\begin{abstract}
Knowledge discovery from data is an inherently iterative process. 
That is, what we know about the data greatly determines our expectations, and therefore, what results we would find interesting and/or surprising. Given new knowledge about the data, our expectations will change. Hence, in order to avoid  redundant results, knowledge discovery algorithms ideally should follow such an iterative updating procedure.

With this in mind, we introduce a well-founded approach for succinctly summarizing data with the most informative itemsets; using a probabilistic maximum entropy model, we iteratively find the itemset that provides us the most novel information---that is, for which the frequency in the data surprises us the most---and in turn we update our model accordingly. As we use the Maximum Entropy principle to obtain unbiased probabilistic models, and only include those itemsets that are most informative with regard to the current model, the summaries we construct are guaranteed to be both descriptive and non-redundant.

The algorithm that we present, called \textsc{mtv}, can either discover the top-\emph{k} most informative itemsets, or we can employ either the Bayesian Information Criterion (\textsc{bic}) or the Minimum Description Length (\textsc{mdl}) principle to automatically identify the set of itemsets that together summarize the data well. In other words, our method will `tell you what you need to know' about the data. Importantly, it is a one-phase algorithm: rather than picking itemsets from a user-provided candidate set, itemsets and their supports are mined on-the-fly. To further its applicability, we provide an efficient method to compute the maximum entropy distribution using Quick Inclusion-Exclusion.

Experiments on our method, using synthetic, benchmark, and real data, show that the discovered summaries are succinct, and correctly identify the key patterns in the data. The models they form attain high likelihoods, and inspection shows that they summarize the data well with increasingly specific, yet non-redundant itemsets.
\end{abstract}




\acmformat{M. Mampaey, J. Vreeken, and N. Tatti, 2011. \ourtitle.}

\begin{bottomstuff}
Michael Mampaey is supported by a Ph.D. grant of the Agency for Innovation by Science and Technology in Flanders (IWT). 
Jilles Vreeken is supported by a Post-Doctoral Fellowship of the Research Foundation -- Flanders (FWO).
Nikolaj Tatti is supported by a Post-Doctoral Fellowship of the Research Foundation -- Flanders (FWO).

Authors' address: M.\ Mampaey, J.\ Vreeken, and N.\ Tatti, ADReM Research Group, Department of Mathematics and Computer Science, University of Antwerp, Middelheimlaan 1, 2020 Antwerp, Belgium.
\end{bottomstuff}

\maketitle

\pagebreak
\input{introduction}
\input{related}

\input{prel}
\input{model}
\input{problem}
\input{algorithm}

\input{experiments}
\input{discussion}

\input{conclusion}

\begin{acks}
The authors wish to thank Tijl De Bie for pre-processing and making available of the ICDM Abstracts data, the i-ICT of Antwerp University Hospital (UZA) for providing the MCADD data and expertise, and last, but not least, the anonymous referees whose insightful comments and suggestions helped to improve the paper considerably.
\end{acks}

\input{appendix}

\bibliographystyle{acmsmall}
\bibliography{abbreviations,bib-jilles}

\received{July 2011}{April 2012}{May 2012}

\end{document}

%% file: defines.tex
\newcommand{\set}[1]{\left\{#1\right\}}

\newcommand{\fpr}[1]{\mathopen{}\left(#1\right)}

\newcommand{\abs}[1]{{\left|#1\right|}}

\newcommand{\enset}[2]{\left\{#1 ,\ldots , #2\right\}}

\newcommand{\funcdef}[3]{{#1}:{#2} \to {#3}}
\newcommand{\score}[1]{s\fpr{#1}}
\newcommand{\bic}{\textrm{\sc bic}}
\newcommand{\bicscore}[1]{\bic\fpr{#1}}
\newcommand{\mdl}{\textrm{\sc mdl}}
\newcommand{\mdlscore}[1]{\mdl\fpr{#1}}
\newcommand{\faai}{\Phi}

\newcommand{\define}{\leftarrow}

\newcommand{\ifam}[1]{\mathcal{#1}}

\newcommand{\freq}[1]{\mathit{fr}\fpr{#1}}
\newcommand{\supp}[1]{\mathit{supp}\fpr{#1}}

\newcommand{\exact}[1]{\mathit{e}\fpr{#1}}
\newcommand{\cum}[1]{\mathit{c}\fpr{#1}}

\newcommand{\sets}[1]{\mathit{sets}\fpr{#1}}
\newcommand{\block}[1]{\mathit{block}\fpr{#1}}
\newcommand{\closure}[1]{\mathit{closure}\fpr{#1}}
\newcommand{\prior}{\mathit{prior}}

\newcommand{\partsize}[1]{\mathit{nb}\fpr{#1}}

\newcommand{\ent}[1]{H\fpr{#1}}
\newcommand{\kl}[2]{\mathit{KL}\fpr{#1 \,\|\, #2}}
\newcommand{\pemp}{p^*}

\newcommand{\trans}{\mathcal{T}}
\newcommand{\db}{D}
\newcommand{\coll}{\ifam{C}}
\newcommand{\bg}{\ifam{B}}

\newcommand{\argmin}{\operatornamewithlimits{arg\ min}}
\newcommand{\argmax}{\operatornamewithlimits{arg\ max}}

\newtheorem{prblm}{Problem}

\newtheorem{Example}[theorem]{Example}
\newtheorem{apxcorollary43}{Corollary~\ref{cor:likelihood}} 
\newtheorem{apxtheorem45}{Theorem~\ref{theorem:4.5}} 
\newtheorem{apxtheorem418}{Theorem~\ref{theorem:itscorrect}} 

\SetKw{AND}{and}
\SetKw{OR}{or}
\SetKw{EACH}{each}

\SetKwInOut{Input}{input}
\SetKwInOut{Output}{output}
\SetArgSty{textnormal}

\mathchardef\mhyphen="2D

%% file: introduction.tex
\section{Introduction}
\label{sec:introduction}

Knowledge discovery from data is an inherently iterative process. 
That is, what we already know about the data greatly determines our expectations, and therefore, which results we would find interesting and/or surprising. Early on in the process of analyzing a database, for instance, we are happy to learn about the generalities underlying the data, while later on we will be more interested in the specifics that build upon these concepts.
Essentially, this process comes down to summarization: we want to know what is interesting in the data, and we want this to be reported as succinctly as possible, without redundancy.

As a simple example, consider supermarket basket analysis. Say we just learned that {\em pasta} and {\em tomatoes} are sold together very often, and that we already knew that many people buy {\em wine}. Then it would not be very interesting to be told that the combination of these three items is also sold frequently; although we might not have been able to predict the sales numbers exactly, our estimate would most likely have come very close, and hence we can say that this pattern is redundant.

At the same time, at this stage of the analysis we are probably also not  interested in highly detailed patterns, e.g., an itemset representing the many ingredients of an elaborate Italian dinner. While the frequency of this itemset may be surprising, the pattern is also highly specific, and may well be better explained by some more general patterns. Still, this itemset might be regarded as highly interesting further on in the discovery process, after we have learned those more general patterns, and if this is the case, we would like it to be reported at that time. 

In a nutshell, that is the approach we adopt in this paper: we incrementally adjust our model as we iteratively discover informative patterns, in order to obtain a non-redundant summary of the data.

As natural as it may seem to update a knowledge model during the discovery process, and in particular to iteratively find results that are informative with regard to what we have learned so far, few pattern mining techniques actually follow such a dynamic approach. That is, while many methods provide a series of patterns in order of interestingness, most score these patterns using a static model; during this process the model, and hence the itemset scores, are not updated with the knowledge gained from previously discovered patterns. 
For instance, \citet{tan:02:selecting} and \citet{geng:06:interestingness} respectively study 21 and 38 interestingness measures, all of which are static, and most of which are based on the independence model \citep[e.g.,][]{brin:97:beyond,aggarwal:98:new}. The static approach inherently gives rise to a typical problem of traditional pattern mining: overwhelmingly large and highly redundant pattern sets. 

Our objective is to mine succinct summaries of binary data, that is, to obtain a small, yet high-quality set of itemsets that describes key characteristics of the data at hand, in order to gain useful insight. 
This is motivated by the fact that many existing algorithms often return too large collections of patterns with considerable redundancy, as discussed above. 
The view that we take in this paper on succinctness and non-redundancy is therefore a fairly strict one. 

While we are not the first to propose a method that updates its scoring model dynamically---examples include the swap randomization-based approach by \citet{hanhijarvi:09:tell} and the compression-based approach by \citet{vreeken:11:krimp}---there are several differences with existing methods. 
For instance, the former requires generating many randomized databases in order to estimate frequencies, whereas our model is analytical, allowing for direct frequency calculation. 
The models for the latter are not probabilistic, and while non-redundant with respect to compression, 
can contain patterns that are variations of the same theme. We treat related work in more detail in Section~\ref{sec:related}, but let us first discuss the basic features of our approach.

To model the data, we use the powerful and versatile class of maximum entropy models. We construct a maximum entropy distribution that allows us to directly calculate the expected frequencies of itemsets. Then, at each iteration, we return the itemset that provides the most information, i.e., 
for which our frequency estimate was most off. We update our model with this new knowledge, and continue the process. The non-redundant model that contains the most important information is thus automatically identified. Therefore, we paraphrase our method as `tell me what I need to know'.

While in general solving the maximum entropy model is infeasible, we show that in our setting it can be solved efficiently---depending on the amount of overlap between the selected patterns. Similarly, we give an efficient method for estimating frequencies from the model. Further, we provide an efficient convex heuristic for effectively pruning the search space when mining the most informative itemsets.
This heuristic allows us to mine collections of candidate itemsets on the fly, instead of picking them from a larger candidate collection that has to be materialized beforehand. 

Our approach does not require user-defined parameters such as a significance level or an error threshold. In practice, however, we allow the user to specify itemset constraints such as a minsup threshold or a maximum size, which reduces the size of the search space. Such constraints can be integrated quite easily into the algorithm.

We formalize the problem of identifying the most informative model both by the Bayesian Information Criterion (\textsc{bic}), and by the Minimum Description Length (\textsc{mdl}) principle; both are well-known and well-understood model selection techniques that have natural interpretations. 
To heuristically approximate these ideal solutions, we introduce the \textsc{mtv} algorithm, for mining the most informative itemsets. Alternatively, by its iterative nature, \textsc{mtv} can also mine the top-\emph{k} most informative itemsets. Finally, our approach easily allows the user to infuse background knowledge into the model (in the form of itemset frequencies, column margins, and/or row margins), to the end that redundancy with regard to what the user already knows can effectively be avoided. 

Experiments on real and synthetic data show that our approach results in succinct, non-redundant data summaries using itemsets, and provide intuitive descriptions of the data. Since they only contain a small number of key patterns about the data, they can easily be inspected manually by the user, and since redundancy is reduced to a minimum, the user knows that every pattern he or she looks at will be informative.

An earlier version of this work appeared as \citet{mampaey:11:tell}. In this work we significantly extend said paper in the following ways. We define an {\sc mdl}-based quality measure for a collection of itemsets, which is more expressive and supersedes the {\sc bic} score presented earlier. To solve the maximum entropy model, we introduce a new and faster algorithm employing Quick Inclusion-Exclusion to efficiently compute the sizes of the transaction blocks in partitions induced by an itemset collection. Moreover, we formulate how basic background information such as column margins and row margins (i.e., the density of each row and column), can be included into the model, such that results following from this background knowledge will not be regarded as informative, since the background knowledge of the user determines what he or she will find interesting. Further, we provide extensive experimental validation of our methods using fourteen real and synthetic datasets covering a wide range of data characteristics.

The remainder of this paper is organized as follows. First, Section~\ref{sec:related} discusses related work. We cover notation and preliminaries in Section~\ref{sec:prel}. Next, in Section~\ref{sec:theory} we give an introduction to Maximum Entropy models and how to compute them efficiently, and show how to measure the interestingness of a set of itemsets using \textsc{bic} and \textsc{mdl}. We give a formal problem statement  in Section~\ref{sec:problem}. Subsequently, we present the \textsc{mtv} algorithm in Section~\ref{sec:algorithm}. In Section~\ref{sec:experiments}  we report on the experimental evaluation of our method. We round up with a discussion in Section~\ref{sec:discussion} and conclude in Section~\ref{sec:conclusion}. For reasons of presentation some proofs have been placed in the Appendix.

%% file: related.tex
\section{Related Work}
\label{sec:related}

Selecting or ranking interesting patterns is a well-studied topic in data mining. Existing techniques can roughly be split into two groups.

\subsection{Static Approaches}
The first group consists of techniques that measure how \emph{surprising} the support of an itemset is compared against some null hypothesis: the more the observed frequency deviates from the expected value, the more interesting it is. 
In frequent itemset mining, for instance, one can consider the null hypothesis to be that no itemset occurs more than the minimum support threshold. 
Similarly, in tile mining~\citep{geerts:04:tiling}, one searches for itemsets with a large area (support multiplied by size); here the underlying assumption is that the area of any itemset is small.
A simple and often-used probabilistic null hypothesis is the independence model~\cite{brin:97:beyond,aggarwal:98:new}. More flexible models have been suggested, for example, Bayesian Networks~\citep{jaroszewicz:04:interestingness}. The major caveat of these approaches is that the null hypothesis is static and hence we keep rediscovering the same information. As a result, this will lead to pattern collections with high levels of redundancy.

Swap randomization was proposed by \citet{gionis:07:assessing} and \citet{hanhijarvi:09:tell} as a means to assess the significance of data mining results through randomization. To this end, \citet{gionis:07:assessing} gave an algorithm by which randomized data samples can be drawn by repeatedly swapping values locally, such that the background knowledge is maintained---essentially, a Markov chain is defined. 
Then, by repeatedly sampling such random datasets, one can assess the statistical significance of a result by calculating empirical p-values. 
While the original proposal only considered row and column margin as background knowledge, \citet{hanhijarvi:09:tell} extended the approach such that cluster structures and itemset frequencies can be maintained.

While a very elegant approach, swap randomization does suffer from some drawbacks. First of all, there are no theoretical results on the mixing time of the Markov chain, and hence one has to rely on heuristics (e.g., swap as many times as there are ones in the data). Second, as typically very many swaps are required to obtain a randomized sample of the data, and finding suitable swaps is nontrivial, the number of randomized datasets we can realistically obtain is limited, and hence so is the p-value resolution by which we measure the significance of results. As our approach is to model the data probabilistically by the Maximum Entropy principle, we do not suffer from convergence issues, and moreover, as our model is analytical in nature, we can calculate exact probabilities and p-values.

\subsection{Dynamic Approaches}
The alternative approach to measuring informativeness statically, is to rank and select itemsets using a dynamic hypothesis. That is, when new knowledge arrives, e.g., in the form of an interesting pattern, the model is updated such that we take this newly discovered information into account, and hence we avoid reporting redundant results.
The method we present in this paper falls into this category.

Besides extending the possibilities for incorporating background knowledge into a static model, the aforementioned approach by \citet{hanhijarvi:09:tell} discusses that by iteratively updating the randomization model, redundancy is eliminated naturally. 

The {\sc mini} algorithm by \citet{gallo:07:mini} also uses row and column margins
to rank itemsets. It first orders all potentially interesting itemsets by computing their p-values according
to these margins. Then, as itemsets are added, the p-values are recomputed, and the 
itemsets re-ordered according to their new p-values. However, this method does not allow querying, and requires a set of candidates to be mined beforehand. 

{\sc Krimp}, by \citet{siebes:06:item,vreeken:11:krimp}, employs the {\sc mdl} principle to select those itemsets that together compress the data best. As such, patterns that essentially describe the same part of the data are rejected. The models it finds are not probabilistic, and cannot straightforwardly be used to calculate probabilities (although \citet{vreeken:07:dgen} showed data strongly resembling the original can be sampled from the resulting code tables). Further, while non-redundant from a compression point of view, many of the patterns it selects are variations of the same theme. 
The reason for this lies in the combination of the covering strategy and  encoding {\sc Krimp} utilizes. The former prefers long itemsets over short ones, while the second assumes independence between all itemsets in the code table. Hence, for {\sc Krimp} it is sometimes cheaper to encode highly specific itemsets with {\em one} relatively long code, than to encode it with {\em multiple} slightly shorter codes, which may lead to the selection of overly specific itemsets. 
Other differences to our method are that {\sc Krimp} considers its candidates in a static order, and that it is unclear how to make it consider background knowledge.
Recent extensions of \textsc{Krimp} include
the \textsc{Slim} algorithm by \citet{smets:12:slim}, which finds good code tables directly from data by iteratively heuristically finding the optimal addition to the code table, and 
the \textsc{Groei} algorithm by \citet{siebes:11:groei}, which identifies the optimal set of \emph{k} itemsets by beam search instead of identifying the optimal set overall.

\subsection{Modelling by Maximum Entropy}
The use of maximum entropy models in pattern mining has been proposed by several authors, e.g., \citep{tatti:08:decomposable,tatti:08:signifsets,wang:06:summaxent,konto:10:sdm,debie:11:dami}.
Discovering itemset collections with good {\sc bic} scores was
suggested by~\citet{tatti:08:decomposable}.
Alternatively, \citet{tatti:10:probably} samples collections and bases the significance of an itemset on its occurrence in the discovered collections.
However, in order to guarantee that
the score can be computed, the authors restrict themselves to 
a particular type of collections: downward closed and decomposable collections of itemsets.

The method of \citet{tatti:08:signifsets} uses local models. That is, to compute the support of an itemset $X$, the method only uses sub-itemsets of $X$, and outputs a p-value. 
Unlike our approach, it requires a threshold to determine whether $X$ is important. 
Relatedly, \citet{webb:10:self} defines itemsets as \emph{self-sufficient}, if their support differs significantly from what can be inferred from their sub- and supersets; therefore such a model is also local. 

\citet{wang:06:summaxent} incrementally build a maximum entropy model by adding itemsets that deviate more than a given error threshold. The approach ranks and adds itemsets in level-wise batches, i.e., first itemsets of size 1, then of size 2, and so on. This will, however, not prevent redundancy within  batches of itemsets.

\citet{debie:11:dami} proposed an alternative to swap-randomization for obtaining randomized datasets, by modeling the whole data by maximum entropy, using row and column sums as background information; besides faster, and more well-founded, unlike swap-randomization this approach does not suffer from convergence issues. Furthermore, by its analytical nature, exact p-values can be calculated. \citet{konto:10:sdm} used the model to analytically define an interestingness measure, Information Ratio, for noisy tiles, by considering both the expected density of a tile, and the complexity of transferring the true tile to the user.

Although both De~Bie and ourselves model data by the Maximum Entropy principle, there exist important differences between the two approaches. The most elementary is that while we regard it a \emph{bag of samples} from a distribution, De~Bie considers the data as a \emph{monolithic} entity. That is, De~Bie models the whole binary matrix, while we construct probabilistic model for individual rows. Informally said, De~Bie considers the location of a row in the matrix important, whereas we do not. Both these approaches have different advantages. While our approach intuitively makes more sense when modeling, say, a supermarket basket dataset, where the data consists of individual, independent samples; the monolithic approach is more suited to model data where the individual rows have meaning, say, for a dataset containing mammal presences for geographic locations. Moreover, while for our models it is straightforward to include itemset frequencies (which does not include specifying transaction identifiers), such as \emph{tomatoes} and \emph{pasta} are sold in 80\% of the transactions, this is currently not possible for the whole-dataset model. 
Additionally, while the De~Bie framework in general allows the background knowledge of a user to be false, in this work we only consider background knowledge consistent with the data. 
As opposed to~\citet{konto:10:sdm}, we do not just rank patterns according to interestingness, but formalize model selection techniques such as {\sc bic} or {\sc mdl} such that we can identify the optimal model, and hence avoid discovering overly complex models.

As overall comments to the methods described above, we note that in contrast to our approach, most of the above methods require the user to set one or several parameters to asses the informativeness of itemsets, e.g., a maximum error threshold or a significance level (although we do allow thresholding the support of candidate itemsets). Many also cannot easily be used to estimate the frequency of an itemset. Further, all of them are two-phase algorithms, i.e., they require that the user first provides a collection of candidate (frequent) itemsets to the algorithm, which must be completely mined and stored first, before the actual algorithm itself is run.

%% file: prel.tex
\section{Preliminaries and Notation}
\label{sec:prel}

This section provides some preliminaries and the notation that we will use throughout the rest of this paper.

By a \emph{transaction} we mean a binary vector of size $N$ generated by some unknown distribution. The $i$th element in a random transaction corresponds to an \emph{attribute} or \emph{item} $a_i$, a Bernoulli random variable.  We denote the set of all items by $A = \enset{a_1}{a_N}$. We denote the set of all possible transactions by $\trans = \set{0, 1}^N$\!.

The input of our method is a \emph{binary dataset} $\db$, which is a bag 
of $\abs{\db}$ transactions. Given the data $D$ we define an empirical distribution
\[
	q_{\db}(a_1 = v_1, \ldots, a_N = v_N) = \abs{\set{t \in \db \mid t = v}} / {\abs{\db}}\;.
\]

An \emph{itemset} $X$ is a subset of $A$. For notational convenience, given a
distribution $p$, an itemset $X = \enset{x_1}{x_L}$, and a binary vector $v$ of
length $L$, we often use $p(X = v)$ to denote $p(x_1 = v_1, \ldots, x_L = v_L)$.
If $v$ consists entirely of 1's, we use the notation $p(X = 1)$.  
A transaction $t$ is said to \emph{support} an itemset $X$ if it has 1's for all attributes that $X$ identifies. As such, the {support} of an itemset $X$ in a database $\db$ is defined as 
\[
	\mathit{supp}(X) = |\set{t \in \db \mid \pi_{x_1}(t)=1,\ldots,\pi_{x_L}(t)=1}|\;,
\]
where $\pi_{x}(t) \in \{0,1\}$ is the projection of transaction $t$ onto item $x$. Analogously, the \emph{frequency} of an itemset $X$ in a dataset $\db$ is defined as
\[
	\freq{X} = \mathit{supp}(X)/|\db| = q_{\db}(X = 1)\;.
\]  

An \emph{indicator function} $\funcdef{S_X}{\trans}{\set{0, 1}}$ of an itemset $X$
maps a transaction $t$ to a binary value such that $S_X(t) = 1$ if and only if
$t$ supports $X$.

The \emph{entropy} of a distribution $p$ over $\trans$ is defined as 
\[
	\ent{p} = - \sum_{t \in \trans} p(A = t) \log p(A = t)\;,
\]
where the base of the logarithm is 2, and by convention $0 \log 0 = 0$.
The entropy of $p$ is the expected number of bits needed to optimally encode a transaction $t$.

Finally, the \emph{Kullback-Leibler divergence} \citep{cover:06:elements} between two distributions $p$ and $q$ over $\trans$ is defined as 
\[
	\kl{p}{q} = \sum_{t\in \ifam{T}} p(A=t) \log \frac{p(A=t)}{q(A=t)}\;.
\]
Intuitively, the KL divergence between two distributions $p$ and $q$ 
is the average number of \emph{extra} bits required to encode 
data generated by $p$ using a coding optimal for $q$. 
Since $p$ defines the optimal coding distribution for this data, on average,
it will always cost extra bits when we encode data generated by $p$ using 
a coding distribution $q\neq p$,
and hence $\kl{p}{q}$ is non-negative. The KL divergence equals 0
if and only if $p$ equals $q$. 

%% file: model.tex
\section{Identifying the Best Summary}
\label{sec:theory}

Our goal is to discover the collection $\coll$ of itemsets and frequencies that is the most informative for a dataset $\db$, while being succinct and as little redundant as possible. 
Here, by informative we mean whether we are able to reliably describe, or predict, the data using these itemsets and their frequencies. 

By non-redundancy we mean that, in terms of frequency, every element of $\coll$ provides significant information for describing the data that cannot be inferred from the rest of the itemset frequencies in $\coll$. 
This is equivalent to requiring that the frequency of an itemset $X \in \coll$ should be surprising with respect to $\coll \setminus X$. 
In other words, we do not want $\coll$ to be unnecessarily complex as a collection, or capture spurious information, since we only want it to contain itemsets that we really need.

Informally, assume that we have a quality score $s(\coll, \db)$ 
which measures the quality of an itemset collection $\coll$ with respect to $\db$. Then our aim is to find that $\coll$ with the highest score $s(\coll, \db)$.
Analogously, if we only want to know $k$ itemsets, we aim to find the collection $\coll$ of size at most $k$, with the highest score $s(\coll, \db)$.

Next, we will detail how we define our models, how we define this score, provide theoretical evidence why it is a good choice, and discuss how to compute it efficiently. 

\begin{Example}
As a running example, assume that we have a binary dataset $\db$ with eight items, $a$ to $h$. 
Furthermore, consider the set of itemsets $\coll = \left\{\mathit{abc}, \mathit{cd}, \mathit{def}\right\}$ with frequencies $0.5$, $0.4$ and $0.8$, respectively.
Assume for the moment that based on $\coll$, our method predicts that the frequency of the itemset $\mathit{agh}$ is $0.19$. 
Now, if we observe in the data that $\freq{\mathit{agh}}=0.18$, then we can safely say that $\mathit{agh}$ is redundant with regard to what we already know, as it does not contribute a lot of novel information, and the slight deviation from the expected value may even be coincidental.
On the other hand, if $\freq{\mathit{agh}}=0.7$, then the frequency of $\mathit{agh}$ is surprising with respect to $\coll$, and hence adding it to $\coll$ would strongly increase the amount of information $\coll$ gives us about the data; in other words $\coll\cup \{\mathit{agh}\}$ provides a substantially improved description of $\db$.
\end{Example}

\subsection{Maximum Entropy Model}

In our approach we make use of maximum entropy models. This is a class of probabilistic models that are identified by the Maximum Entropy principle~\citep{csiszar:75:i-divergence,jaynes:82:rationale}. This principle states that the best probabilistic model is the model that makes optimal use of the provided information, and that is fully unbiased (i.e., fully random, or, maximally entropic) otherwise. 
This property makes these models very suited for identifying informative patterns: by using maximum entropy models to measure the quality of a set of patterns, we know that our measurement only relies on the information we provide it, and that it will not be thrown off due to some spurious structure in the data.
These models have a number of theoretically appealing properties, which we will discuss after a formal introduction.

Assume that we are given a collection of itemsets and corresponding frequencies
\begin{equation}
	\langle \coll, \faai\rangle = \langle\enset{X_1}{X_k}, \enset{f_{1}}{f_{k}}\rangle\;,
\end{equation}
 where $X_{i}\subseteq A$ and $f_{i} \in [0,1]$, for $i=1,\ldots,k$. Note that we do not require that the frequencies $f_{i}$ of the itemsets are equal to the frequencies $\freq{X_{i}}$ in the data. If this does hold, we will call $\langle\coll,\faai\rangle$ \emph{consistent with} $\db$. For notational convenience, we will at times omit writing the frequencies $\faai$, and simply use $\coll=\enset{X_{1}}{X_{k}}$, especially when it is clear what the corresponding frequencies are.
Now, we consider those distributions over the set of all transactions $\trans$ that satisfy the constraints imposed by $\langle\coll,\faai\rangle$. That is, we consider the following set of distributions
\begin{equation}\label{eq:alldistributions}
	\ifam{P}_{\langle\coll,\faai\rangle} = \set{p \mid p(X_i = 1) = f_{i}, \mathrm{\ for\ } i=1,\ldots, k}\;.
\end{equation}
In the case that $\ifam{P}_{\langle\coll,\faai\rangle}$ is empty, we call $\langle\coll,\faai\rangle$ \emph{inconsistent}.
Among these distributions we are interested in only one, namely the unique distribution
maximizing the entropy
\begin{equation*}
	\pemp_{\langle\coll,\faai\rangle} = \argmax_{p \in \ifam{P}_{\langle\coll,\faai\rangle}} \ent{p}\;.
\end{equation*}
Again, for notational convenience we will often simply write $\pemp_{\coll}$,
or even omit $\langle\coll,\faai\rangle$ altogether, especially when they are clear from the context. 

The following famous theorem states that the maximum entropy model has an exponential form. This form will help us to discover the model and will be useful to compute the quality score of a model.

\begin{theorem}[Theorem~3.1 in~\citep{csiszar:75:i-divergence}]
\label{thr:exponential}
Given a collection of itemsets and frequencies $\langle \coll, \faai\rangle = \langle\enset{X_1}{X_k}, \enset{f_{1}}{f_{k}}\rangle$, let  $\mathcal{P}_{\langle \coll, \faai\rangle}$ be the set of distributions as defined in Eq.~\ref{eq:alldistributions}. If there is a distribution in $\mathcal{P}_{\langle \coll, \faai\rangle}$ that has only nonzero entries, then the maximum entropy distribution $\pemp_{\langle \coll, \faai\rangle}$ can be written as
\begin{equation} \label{eq:loglinearmodel}
	\pemp_{\langle \coll, \faai\rangle}\fpr{A = t} = u_0 \prod_{X \in \coll} u_X^{S_X(t)}\,,
\end{equation}
where $u_X \in \mathbb{R}$, and $u_0$ is a normalization factor such that $\pemp_{\langle \coll, \faai\rangle}$ is a proper distribution.
\end{theorem}

Note that the normalization factor $u_{0}$ can be thought of as corresponding to the constraint that the empty itemset $\emptyset$ should have a frequency $f_{\emptyset}=1$.
Theorem~\ref{thr:exponential} has a technical requirement that $\mathcal{P}$
needs to contain a distribution with nonzero entries. The easiest way to
achieve this is to apply a Bayesian shift~\cite{gelman:04:bayesian} by redefining the frequencies
$f_{i}' = (1 - \epsilon)f_{i} + \epsilon 2^{-\abs{X_{i}}}$ for some small $\epsilon>0$. 
This way the frequencies remain approximately the same, but no transaction has zero probability.

\subsection{Identifying the Best Model}
Here we describe how we can quantify the goodness of a pattern collection.

A natural first choice would be to directly measure the goodness of fit, using the log-likelihood of the maximum entropy model, that is, 
\begin{equation*}
	\log \pemp_{\langle\coll,\faai\rangle}\fpr{\db} = \log \prod_{t \in \db} \pemp_{\langle\coll,\faai\rangle}\fpr{A = t} = \sum_{t \in \db} \log \pemp_{\langle\coll,\faai\rangle}\fpr{A = t}\;.
\end{equation*}
Note that if $\langle \coll, \faai\rangle$ is inconsistent, $\pemp$ does not exist. In this case we define the likelihood to be zero, and hence the log-likelihood to be $-\infty$.

The following corollary shows that for exponential models, we can easily calculate the log-likelihood.
\begin{corollary}[of Theorem~\ref{thr:exponential}]
\label{cor:likelihood}
The log-likelihood of the maximum entropy distribution $\pemp_{\langle \coll, \faai\rangle}$ for a collection of itemsets and frequencies $\langle \coll, \faai\rangle$ is equal to
\begin{align*}
	\log \pemp_{\langle \coll, \faai\rangle}\fpr{\db}  &= \abs{\db}\bigg(\log u_0 + \sum_{(X_{i}, f_{i}) \in \langle \coll, \faai\rangle} f_{i} \log u_{X_{i}}\bigg)\;.
\end{align*}
\end{corollary}
\begin{proof}
See the Appendix.
\end{proof}
Thus, to calculate the log-likelihood of a collection $\langle \coll, \faai\rangle$, it suffices to compute the parameters $u_X$ and $u_0$ of the corresponding distribution $\pemp_{\langle \coll, \faai\rangle}$. 

The following theorem states that if we are searching for collections with a high likelihood, we can restrict ourselves to collections that are consistent with the data.
\begin{theorem}
\label{th:maxlikelihood}
For a fixed collection of itemsets $\coll=\enset{X_{1}}{X_{k}}$, the likelihood $\pemp_{\langle \coll, \faai\rangle}(\db)$ is maximized if and only if $\langle \coll,\faai\rangle$ is consistent with $\db$, that is, $f_{i}=\freq{X_{i}}$ for all $i=1\ldots,k$.
\end{theorem}

For our goal, maximum entropy models are theoretically superior over any other  model. Let us discuss why. Let $\db_1$ and $\db_2$ be two datasets such that $\freq{X \mid \db_1}
= \freq{X \mid \db_2}$ for any $X \in \coll$. Let $\pemp_1$ and $\pemp_2$ be
the corresponding maximum entropy models, then, by definition, $\pemp_1 =
\pemp_2$.  In other words, the model depends only the support of the chosen
itemsets.  This is a natural requirement, since we wish to measure the quality
of the statistics in $\coll$ and nothing else. Similarly,
Corollary~\ref{cor:likelihood} implies that $\pemp_1\fpr{\db_2} = \pemp_1\fpr{\db_1}$. This is also a natural property because otherwise, the score would be depending on some statistic not included in $\coll$. Informally said, the scores are equal if we cannot distinguish between $\db_1$ and $\db_2$ using the information $\coll$ provides us.
The next theorem states that among all such models, the maximum entropy model has the best likelihood, in other words, the maximum entropy model uses the available information as efficiently as possible.

\begin{theorem}
\label{theorem:4.5}
Assume a collection of itemsets $\coll = \enset{X_1}{X_k}$ and let $\pemp_\coll$ be the maximum entropy model,
computed from a given dataset $\db$.
Assume also an alternative model $r(A = t \mid f_1, \ldots, f_k)$, where $f_i = \freq{X_i \mid D}$, 
that is, a statistical model parametrized by the frequencies of $\coll$. Assume
that for any two datasets $\db_1$ and $\db_2$, where $\freq{X \mid \db_1}
= \freq{X \mid \db_2}$ for any $X \in \coll$, it holds that
\[
	1/\abs{D_1} \log r(D_1 \mid f_1, \ldots, f_k) = 1/\abs{D_2} \log r(D_2 \mid f_1, \ldots, f_k)\;.
\]
Then $\pemp_{\coll}(\db) \geq r(\db)$ for any dataset $\db$.
\end{theorem}

\begin{proof}
See the Appendix.
\end{proof}

Using only log-likelihood to evaluate a model, however, suffers from overfitting: larger collections of itemsets will
always provide more information, hence allow for better estimates, and therefore have a better log-likelihood.
Consequently, we need to prevent our method from overfitting. 
In order to do so, we will explore the Bayesian Information Criterion (\textsc{bic}), and the Minimum Description Length (\textsc{mdl}) principle---both of which are well-known and well-founded model selection techniques. We start by discussing \textsc{bic}, which is the least strict, and least involved of the two.

The Bayesian Information Criterion ({\sc bic}) measures the quality of a model by taking both its log-likelihood, and the number of parameters of said model into account. It favors models that fit the data well using few parameters; in our setting, the number of parameters of the model $\pemp$ corresponds exactly to the number of itemsets $k$. It has a strong theoretical support in Bayesian model selection~\citep{schwarz:78:estimating}.

\begin{definition}
Given a collection of itemsets $\langle\coll,\faai\rangle=\langle\{X_{1},\ldots,X_{k}\}, \{f_{1},\ldots,f_{k}\}\rangle$, the {\sc bic} score with respect to a dataset $\db$ is defined as
\begin{equation}
	\bicscore{\langle\coll,\faai\rangle, \db} = - \log \pemp_{\langle\coll,\faai\rangle}\fpr{\db} + k/2 \log \abs{\db}\;.
\end{equation}
\end{definition}

The better a model fits the data, the higher its likelihood. On the other hand, the more parameters the model has---i.e., the more complex it is---the higher its penalty. Therefore, it is be possible that a model that fits the data slightly worse, but contains few parameters, is favored by {\sc bic} over a model that fits the data better, but is also more complex. From Corollary~\ref{cor:likelihood} we see that the first term of the {\sc bic} score is equal to $\abs{\db} H(\pemp_{\langle\coll,\faai\rangle})$. Hence, the likelihood term grows faster than the penalty term with respect to the size of the $\db$. As such, the more data (or evidence) we have, the more complicated the model is allowed to be.

\begin{corollary}[of Theorem~\ref{th:maxlikelihood}]
For a fixed collection of itemsets $\coll=\enset{X_{1}}{X_{k}}$, the {\sc bic} score $\bicscore{\langle\coll,\faai\rangle, \db}$ is minimized if and only if $\langle\coll,\faai\rangle$ is consistent with $\db$, that is, $f_{i}=\freq{X_{i}}$ for all $i$.
\end{corollary}
\begin{proof}
Follows directly from Theorem~\ref{th:maxlikelihood} and the fact that the {\sc bic} penalty term, $k/2\log\abs{\db}$, does not depend on the frequencies $f_{i}$.
\end{proof}

While the {\sc bic} score helps to avoid overfitting, it is somewhat simplistic. That is, it only incorporates the \emph{number} of itemsets to penalize a summary, and not their complexity. As stated in the introduction, if possible, we would typically rather be given some number of general patterns than the same number of highly involved patterns. 
\textsc{mdl} provides us a means to define a score that also takes into account the complexity of the itemsets in $\coll$.

The Minimum Description Length (\textsc{mdl}) principle \citep{rissanen:78:mdl,grunwald:05:tut}, like its close cousin \textsc{mml} (Minimum Message Length)~\citep{wallace:05:mml}, is a practical version of Kolmogorov Complexity~\citep{vitanyi:93:book}. All three embrace the slogan {\em Induction by Compression}. The \textsc{mdl} principle can be roughly described as follows.

Given a dataset $\db$ and a set of models $\mathcal{M}$ for $\db$, the best model $M \in \mathcal{M}$ is the one that minimizes
\[
L(M) + L(\db \mid M)
\]
in which
\begin{itemize}
\item $L(M)$ is the length, in bits, of the description of the model $M$, and
\item $L(\db \mid M)$ is the length, in bits, of the description of the
      data, encoded with $M$.
\end{itemize}

This is called two-part \textsc{mdl}, or {\em crude} \textsc{mdl}. This stands opposed to
{\em refined} \textsc{mdl}, where model and data are encoded together \citep{grunwald:07:book}.
We use two-part \textsc{mdl} because we are specifically interested in the model: the set of itemsets that yields the best description length. Further, although refined \textsc{mdl} has stronger theoretical foundations, it cannot be computed except for some special cases. We should also point out that refined \textsc{mdl} is asymptotically equivalent to \textsc{bic} if the number
of transactions goes to infinity and the number of free parameters stays
fixed. However, for moderate numbers of transactions there may be significant
differences. Generally speaking, \textsc{mdl} tends to be  more conservative than \textsc{bic}~\cite{grunwald:07:book}.

To use \textsc{mdl}, we have to define what our set of models $\mathcal{M}$ is, how a model $M$ describes a database, and how all of this is encoded in bits. 
Intuitively, we want to favor itemset collections that are small, i.e., collections which can describe the data well, using few itemsets. At the same time, we also prefer collections with small itemsets over collections with large ones.
\begin{definition}
\label{def:mdlscore}
Given a collection of itemsets $\langle\coll,\faai\rangle=\langle\{X_{1},\ldots,X_{k}\}, \{f_{1},\ldots,f_{k}\}\rangle$, let $x=\sum_{i=1}^{k}\abs{X_{i}}$.
We define the {\sc mdl} score of $\langle\coll,\faai\rangle$ with respect to the dataset $\db$ as
\begin{equation} \label{eq:mdlscore}
	\mdlscore{\langle\coll,\faai\rangle, \db} = L(\db \mid \langle\coll,\faai\rangle) + L(\langle\coll,\faai\rangle)\;,
\end{equation}
where 
\[
L(\db \mid \langle\coll,\faai\rangle) = -\log \pemp_{\langle\coll,\faai\rangle}\fpr{\db} \quad \textrm{and} \quad L(\langle\coll,\faai\rangle) = l_{1} k + l_{2} x + 1\;,
\]
with 
\[
	l_{1} = \log\abs{\db} + N\log(1+N^{-1})+ 1 \approx \log\abs{\db} + \log e + 1
\]
and
\[
	l_{2} = \log N\;.
\]
Whenever $\db$ is clear form the context, we simply write $\mdlscore{\langle\coll,\faai\rangle}$.
\end{definition}
The first term is simply the negative log-likelihood of the model, which corresponds to the description length of the data given the maximum entropy model induced by $\coll$. The second part is a penalty term, which corresponds to the description length of the model. It is a linear function of $k=\abs{\coll}$ and $x=\sum_{i}\abs{X_{i}}$, of which the coefficients depend on $N$ and $\abs{\db}$. How it is derived is explained further below.
The smaller this score, the better the model. Given two collections with an equal amount of itemsets, the one containing fewer items is penalized less; conversely, if they have the same total number of items, the one that contains those items in fewer itemsets is favored.
Consequently, the best model is identified as the model that provides a good balance between high likelihood and low complexity. Moreover, we automatically avoid redundancy, since models with redundant itemsets are penalized for being too complex, without sufficiently improving the likelihood.

With this quality score we evaluate collections of itemsets, rather than the (maximum entropy) distributions we construct from them. The reason for this is that we want to summarize the data with a succinct set of itemsets, not model it with a distribution. A single distribution, after all, may be described by many different collections of itemsets, simple or complex. Further, we assume that the set $\ifam{M}$ of models consists of collections of itemsets which are represented as vectors, rather than as sets. This choice keeps the quality score function computationally simple and intuitive, and is not disadvantageous: if $\coll$ contains duplicates, they simply increase the penalty term. Additionally, we impose no restrictions on the consistency of $\langle\coll,\faai\rangle$, that is, there are collections $\langle\coll,\faai\rangle$ for which $\ifam{P}_{\langle\coll,\faai\rangle}$ is empty, and hence the maximum entropy distribution does not exist. As mentioned above, in this case we define the likelihood to be zero, and hence the description length is infinite.

We now describe the derivation of the penalty term, which equals 
\[
	k \log\abs{\db} + x \log N + (k + 1) + kN\log(1+N^{-1})\;.
\]
To describe an itemset we encode a support using $\log \abs{\db}$ bits and the
actual items in the itemsets using $\log N$ bits. This gives us the first two
terms. We use the third term to express whether there are more itemsets,
one bit after each itemset\footnote{Although it is intuitive, and it provides good results in practice, using $1$ bit to signal the end of an itemset is slightly inefficient; it could be further optimized by the decision tree encoding of~\citet{wallace:93:trees}.}. and one extra bit to accommodate the case
$\coll=\emptyset$. The term $kN\log (1+N^{-1})$ is a normalization factor, to ensure that the encoding is \emph{optimal} for the prior distribution over all pattern collections. That is, the encoding corresponds to a distribution
\[
	\prior(\langle\coll,\faai\rangle)=2^{- l_{1} k - l_{2} x - 1}\;,
\]
which assigns high probability to simple summaries, and low probability to complex ones. 
The following equation shows that the above encoding is optimal for this prior.
\begin{align*}
	\sum_{\langle\coll,\faai\rangle} \prior(\langle\coll,\faai\rangle)&= \sum_{k=0}^{\infty}\sum_{x=0}^{kN} \binom{kN}{x}\abs{\db}^{k} 2^{-l_{1} k - l_{2} x - 1} \\
	&= \sum_{k=0}^{\infty} 2^{- k - 1} 2^{- kN \log(1+N^{-1})} 2^{ - k\log\abs{D} } \abs{\db}^{k} \sum_{x=0}^{kN} \binom{kN}{x} N^{-x}\\
	&= \sum_{k=0}^{\infty} 2^{- k - 1} 2^{- kN \log(1+N^{-1})} (1+N^{-1})^{kN}\\
	&=\sum_{k=0}^{\infty} 2^{-k-1} = 1
\end{align*}

The following corollary shows that for identifying the \textsc{mdl} optimal model, it suffices to only consider summaries that are consistent with the data.
\begin{corollary}[of Theorem~\ref{th:maxlikelihood}]
For a fixed collection of itemsets $\coll=\enset{X_{1}}{X_{k}}$, the {\sc mdl} score $\mdlscore{\langle\coll,\faai\rangle, \db}$ is maximized if and only if $\langle\coll,\faai\rangle$ is consistent with $\db$, that is, $f_{i}=\freq{X_{i}}$ for all $i$.
\end{corollary}
\begin{proof}
Follows directly from Theorem~\ref{th:maxlikelihood} and the fact that for a fixed $\ifam{C}$, the frequencies $f_{i}$ are encoded with a constant length $\log\abs{\db}$, and hence the penalty term is always the same.
\end{proof}
Therefore, in the rest of this paper we will assume that $\langle\coll,\faai\rangle$ is always consistent with $\db$, and hence we will omit $\faai$ from notation.

\subsection{Reducing Redundancy}

Here we show that our score favors collections of itemsets that exhibit low redundancy, and make a theoretical link with some popular lossless redundancy reduction techniques from the pattern mining literature.
Informally, we define redundancy as anything that does not deviate (much) from our expectation, or in other words is unsurprising given the information that we already have.
The results below hold for {\sc mdl} as well as for {\sc bic}, and hence we write $s$ to denote either score.

A baseline technique for ranking itemsets is to compare the observed frequency
against the expected value of some null hypothesis. The next theorem shows that
if the observed frequency of an itemset $X$ agrees with the expected value $\pemp(X
= 1)$, then $X$ is redundant.

\begin{theorem}
\label{thr:deviate}
Let $\coll$ be a collection of itemsets and let $\pemp$ be the corresponding maximum entropy model. Let $X \notin \coll$ be an itemset such that $\freq{X} = \pemp\fpr{X = 1}$. Then $\score{\coll \cup \set{X}, \db} > \score{\coll, \db}$.
\end{theorem}

\begin{proof}
We will prove the theorem by showing that the likelihood terms for both
collections are equal.  Define the collection $\coll_1 = \coll \cup
\set{X}$ and let $\mathcal{P}_1$ be the corresponding set of
distributions. Let $\pemp_1$ be the distribution maximizing the entropy in
$\mathcal{P}_1$. Note that since $\coll \subset \coll_1$, we have
$\mathcal{P}_1 \subseteq \mathcal{P}$ and hence $\ent{\pemp_1} \leq
\ent{\pemp}$. On the other hand, the assumption in the theorem implies that $\pemp
\in \mathcal{P}_1$ and so $\ent{\pemp} \leq \ent{\pemp_1}$. Thus, $\ent{\pemp}
= \ent{\pemp_1}$ and since the distribution maximizing the entropy is unique,
we have $\pemp = \pemp_1$. This shows that the likelihood terms in
$\score{\coll,\db}$ and $\score{\coll_1,\db}$ are equal. The penalty term
in the latter is larger, which concludes the proof.
\end{proof}

Theorem~\ref{thr:deviate} states that adding an itemset $X$ to $\coll$ improves the
score only if its observed frequency deviates from the expected value. 
The amount of deviation required to lower the score, is determined by the penalty term.
This gives us a convenient advantage over methods that are based solely on deviation,
since they require a user-specified threshold.

Two useful corollaries follow from Theorem~\ref{thr:deviate}, which connect our approach to well-known techniques for removing redundancy from pattern set collections---so-called \emph{condensed representations}.
The first corollary relates our approach to \emph{closed} itemsets~\citep{pasquier:99:discovering}, and \emph{generator} itemsets (also known as \emph{free} itemsets~\citep{boulicaut:03:free}). An itemset is closed if all of its proper supersets have a strictly lower support. An itemset is called a generator if all of its proper subsets have a strictly higher support.

\begin{corollary}[of Theorem~\ref{thr:deviate}]
\label{cor:closed}
Let $\coll$ be a collection of itemsets. Assume that $X, Y \in \coll$ such that $X \subset Y$ and $\freq{X} = \freq{Y} \neq 0$. Assume that $Z \notin \coll$ such that $X \subset Z \subset Y$. Then $\score{\coll \cup \set{Z},\db} > \score{\coll,\db}$.
\end{corollary}

\begin{proof}
Let $p \in \mathcal{P}$, as defined in Eq.~\ref{eq:alldistributions}. We have that $p(X = 1) = \freq{X} = \freq{Y} = p(Y = 1)$.
Hence we must have $p(Z = 1)=\freq{Z}$. Since $\pemp \in \mathcal{P}$, it must hold that $\pemp(Z = 1)=\freq{Z}$.
The result follows from Theorem~\ref{thr:deviate}.
\end{proof}

Corollary~\ref{cor:closed} implies that if the closure and a generator of an itemset $Z$ are already
in the collection, then adding $Z$ will worsen the score. 
The second corollary
provides a similar relation with \emph{non-derivable} itemsets~\citep{calders:07:ndidami}. 
An itemset is called derivable if its support can be inferred exactly given the supports of all of its proper subsets.

\begin{corollary}[of Theorem~\ref{thr:deviate}]
\label{cor:ndi}
Let $\coll$ be a collection of itemsets. Assume that $X \notin \coll$ is a derivable itemset and all proper sub-itemsets of $X$ are included in $\coll$. Then $\score{\coll \cup \set{X},\db} > \score{\coll,\db}$.
\end{corollary}
\begin{proof}
The proof of this corollary is similar to the proof of Corollary~\ref{cor:closed}. 
\end{proof}

An advantage of our method is that it can avoid redundancy in a very general way. 
The closed and non-derivable itemsets are two types of lossless representations,
whereas our method additionally can give us lossy redundancy removal.
For example, in the context of Corollary~\ref{cor:closed}, we can choose to reject $X$ 
from $\coll$ even if $\freq{X}$ does not equal $\freq{Y}$ exactly, and as such
we can prune redundancy more aggressively.

\subsection{Efficiently Computing the Maximum Entropy Model}
\label{sec:modelcompute}
Computing the maximum entropy model comes down to
finding the $u_{X}$ and $u_{0}$ parameters from Theorem~\ref{thr:exponential}.
To achieve this, we use the well-known Iterative Scaling procedure by \citet{darroch:72:generalized}, 
which is given here as Algorithm~\ref{algo:iterscale}. 
Simply put, it iteratively updates the parameters of an exponential distribution, until it converges to the maximum entropy 
distribution $\pemp$ which satisfies a given set of constraints---itemset frequencies in our case.
The distribution is initialized with the uniform distribution, which is done by
setting the $u_{X}$ parameters to $1$, and $u_{0}=2^{-N}$ to properly normalize the distribution. 
Then, for each itemset $X\in\coll$, we adjust the corresponding parameter $u_{X}$ to enforce $p(X = 1)=\freq{X}$ (line~\ref{algo:paramupdate},\ref{algo:normalize}).
This process is repeated in a round robin fashion until $p$ converges, and it can be shown (see \citet{darroch:72:generalized})
that $p$ always converges to the maximum entropy distribution $\pemp$\!. 
Typically the number of iterations required for convergence is low (usually $<$10 in our experiments).

\begin{algorithm}[h]
\DontPrintSemicolon
\Input{itemset collection $\coll=\{X_{1}, \ldots, X_{k}\}$, frequencies $\freq{X_{1}}, \ldots, \freq{X_{k}}$ } 
\Output{parameters $u_{X}$ and $u_{0}$ of the maximum entropy distribution $\pemp_{\coll}$ satisfying $\pemp_{\coll}(X_{i})=\freq{X_{i}}$ for all $i$}

initialize $p$\; \nllabel{algo:inititerscale}
\Repeat{$p$ converges} {
	\For {\EACH $X$ in $\coll$} {
		compute $p(X=1)$\; \nllabel{algo:inferset}
		$u_{X} \define u_{X} \frac{\freq{X}}{p(X = 1)} \frac{1-p(X = 1)}{1-\freq{X}}$\; \nllabel{algo:paramupdate}
		$u_{0} \define u_{0} \frac{1-\freq{X}}{1-p(X = 1)}$\; \nllabel{algo:normalize}
	}
}
\Return {$p$}\;
\caption{\textsc{IterativeScaling}($\coll$)
}
\label{algo:iterscale}
\end{algorithm}

\begin{Example}
In our running example, with $\coll = \{\mathit{abc}, \mathit{cd}, \mathit{def}\}$, 
the maximum entropy model has three parameters $u_{1}, u_{2}, u_{3}$, 
and a normalization factor $u_{0}$. Initially we set $u_{1}=u_{2}=u_{3}=1$
and $u_{0}=2^{-N}=2^{-8}$. Then, we iteratively loop over the itemsets and scale
the parameters. For instance, for the first itemset $\mathit{abc}$ with frequency $0.5$,
we first compute its current estimate to be $2^{-3}=0.125$. 
Thus, we update the first parameter $u_{1}= 1 \cdot (0.5 / 2^{-3}) \cdot ((1-2^{-3}) / 0.5) = 7$. 
The normalization factor becomes $u_{0}=2^{-8} \cdot 0.5 /(1-2^{-3}) \approx 2.2 \cdot 10 ^{-3}$.
Next, we do the same for $\mathit{cd}$, and so on. After a few iterations, the model parameters converge to
$u_{1}=28.5$, $u_{2}=0.12$, $u_{3}=85.4$ and $u_{0}=3\cdot10^{-4}$.
 \end{Example}

Straightforward as this procedure may be, the greatest computational bottleneck is the inference of the probability of an itemset on line~\ref{algo:inferset} of the algorithm,
\begin{equation}
\label{eq:inference1}
	p (X = 1) = \sum_{t \in \trans \atop S_{X}(t) = 1} p(A = t)\;.
\end{equation}
Since this sum ranges over all possible transactions supporting $X$,
it is infeasible to calculate by brute force, even for a moderate number $N$ items.
In fact, it has been shown that querying the maximum entropy model is \textbf{PP}-hard in general~\citep{tatti:06:computational}.

Therefore, in order to be able to query the model efficiently, we introduce a partitioning scheme, 
which makes use of the observation that many transactions have the same probability in the maximum entropy distribution. 
Remark that an itemset collection $\coll$ partitions $\trans$ into blocks of transactions
which support the same set of itemsets. That is, two transactions $t_{1}$ and $t_{2}$ belong 
to the same block $T$ if and only if $S_{X}(t_{1})=S_{X}(t_{2})$ for all $X$ in $\coll$. 
Therefore, we know that $p(A=t_{1})=p(A=t_{2})$ if $p$ is of the form in Eq.~\ref{eq:loglinearmodel}.
This property allows us to define $S_X(T) = S_X(t)$ for any $t \in T$ and $X \in \coll$.
We denote the partition of $\trans$ induced by $\coll$ as $\trans_{\coll}$.
Then the probability of an itemset is
\begin{equation}
\label{eq:inference2}
	p(X = 1) = \sum_{T \in \trans_{\coll} \atop S_{X}(T)=1} p(A \in T)\;.
\end{equation}
The sum in Eq.~\ref{eq:inference1} has been reduced to a sum over blocks of transactions, and
the inference problem has been moved from the transaction space $\trans$ to the block space $\trans_{\coll}$.
In our setting we will see that $|\trans_{\coll}| \ll |\trans|$, which makes inference a lot more feasible. 
In the worst case, this partition may contain up to $2^{|\coll|}$ blocks, however,
through the interplay of the itemsets, it can be as low as $|\coll|+1$.
As explained below, we can exploit, or even choose to limit, the structure of $\coll$, such that practical computation is guaranteed.

All we must do now is obtain the block probabilities $p(A \in T)$.  Since all transactions $t$ in a block $T$ have the same probability
\[
p(A=t)=u_{0} \prod_{X \in \coll} u_{X}^{S_{X}(t)}\;,
\]
it suffices to compute the number of transactions in $T$ to get $p(A \in T)$.
So, let us define $\exact{T}$ to be the number of transactions in $T$, then 
\[
	p(A \in T) = \sum_{t \in T} p(A = t) = \exact{T} u_{0} \prod_{X \in \coll} u_{X}^{S_{X}(T)}\;.
\]

\begin{algorithm}[tb!]
\DontPrintSemicolon
\Input{itemset collection $\coll=\{X_{1}, \ldots, X_{k}\}$}
\Output{block sizes $e(T)$ for each $T$ in $\trans_{\coll}$}

\For{$T$ in $\mathcal{T_{\coll}}$} {
	$I \define \bigcup \{X \mid X \in \sets{T ; \coll}$\}\;
	$\cum{T} \define 2^{N - \abs{I}}$\; \nllabel{algo:cumuprob}
}
sort the blocks in $\mathcal{T_{\coll}}$\; \nllabel{algo:sortblocks}
\For {$T_i$ in $\mathcal{T_{\coll}}$} {
	$\exact{T_i} \define \cum{T_{i}}$\;
	\For {$T_j$ in $\mathcal{T_{\coll}}$, with $j < i$} {
		\If{$T_i \subset T_j$} {
			$\exact{T_i} \define \exact{T_i} - \exact{T_j}$\;
		}
	}
}

\Return $\trans_{\coll}$\;

\caption{{\sc ComputeBlockSizes}($\coll$)
}
\label{algo:blocksizes}
\end{algorithm}

Algorithm~\ref{algo:blocksizes} describes \textsc{ComputeBlockSizes}, as given  in \citep{mampaey:11:tell}.
In order to compute the block sizes $\exact{T}$, we introduce a partial order on $\trans_{\coll}$.
Let 
\[
	\sets{T ; \coll} = \set{X \in \coll\mid S_{X}(T)=1}
\]
be the itemsets of $\coll$ that occur in the transactions of $T$. 
Note that every block corresponds to a unique subset of $\coll$; conversely a subset of $\coll$ either corresponds to an empty block of transactions, or to a unique nonempty transaction block.
We can now define the partial order on $\trans_{\coll}$ as follows, 
\[
T_{1} \subseteq T_{2}\quad\mathrm{if\ and\ only\ if}\quad\sets{T_1 ; \coll} \subseteq \sets{T_2; \coll}\;.
\] 
In order to compute the size $\exact{T}$ of a block, we start from its \emph{cumulative} size,
\begin{equation}
\cum{T} = \sum_{T' \supseteq T} \exact{T'}\;,
\end{equation}
which is the number of transactions that contain \emph{at least} all the itemsets in $\sets{T ; \coll}$. 
For a given block $T$, let $I = \bigcup\{X \mid X\in\sets{T ; \coll}\}$.
That is, $I$ are the items that occur in all transactions of $T$.
Then it holds that $c(T) = 2^{N - |I|}$, where $N$ is the total number of items.
To obtain the block sizes $\exact{T}$ from the cumulative sizes $\cum{T}$, we use the
Inclusion-Exclusion principle. 
To that end, the blocks are topologically sorted such that if $T_{2} \subset T_{1}$, then $T_{1}$ occurs before $T_{2}$.
The algorithm then reversely iterates over the blocks in a double loop,
subtracting block sizes, using the identity 
\begin{equation}
\exact{T} = \cum{T} - \sum_{T' \supsetneq T} \exact{T'}\;.
\end{equation}

\begin{Example}
\label{example:regularsubtract}
Assume again that we have a dataset with eight items ($a$ to $h$\hspace{0.1em}), and an itemset collection containing 
three itemsets $\coll = \left\{\mathit{abc}, \mathit{cd}, \mathit{def} \right\}$ with frequencies $0.5$, $0.4$ and $0.8$, respectively. 

Table~\ref{fig:blocks} shows the sizes of the transaction blocks. Note that while there are 256 transactions in $\mathcal{T}$,
there are only 7 blocks in $\trans_{\coll}$, whose sizes and probabilities are to be computed (the eighth combination 
`$\mathit{abc}$ and $\mathit{def}$ but not $\mathit{cd}$' is clearly impossible).

Let us compute the size of the first three blocks.
For the first block, $I=\mathit{abcdef}$ and therefore $\cum{T}=4$, for the second block $I=\mathit{abcd}$, and for the third block $I=\mathit{abc}$.
Since the first block is the maximum with respect to the order $\subseteq$, its cumulative size is simply its size, so $\exact{T}=4$.
For the second block, we subtract the first block, and obtain $\exact{T}=16-4=12$. 
From the third block we subtract the first two blocks, and we have $\exact{T}=32-12-4=16$.
Now, to compute, say, $p(\mathit{abc}=1)$, we simply need the sizes of the blocks containing $abc$,
and the current model parameters,
\[
p(\mathit{abc}=1) = 4 (u_{0} u_{1} u_{2} u_{3}) + 12 (u_{0} u_{1} u_{2}) + 16 (u_{0} u_{1})\;.
\]

\begin{table}[]
\caption{Transaction blocks for the running example above, with $X_{1}=\mathit{abc}$, $X_{2}=\mathit{cd}$, and $X_{3}=\mathit{def}$.}
\label{fig:blocks}
\centering
\begin{tabular}{ccccrrcr}
\toprule
$X_{1}$ & $X_{2}$ & $X_{3}$ && $c(T)$ & $e(T)$ && $p(A=t)$\\
\midrule
1&1&1 && 4 & 4 && $u_{0} u_{1} u_{2} u_{3}$ \\
1&1&0 && 16 & 12 && $u_{0} u_{1} u_{2}$ \\
1&0&0 && 32 & 16 && $u_{0} u_{1}$ \\
0&1&1 && 16 & 12 && $u_{0} u_{2} u_{3}$ \\
0&1&0 && 64 & 36 && $u_{0} u_{2}$ \\
0&0&1 && 32 & 16 && $u_{0} u_{3}$ \\
0&0&0 && 256 & 160 && $u_{0}$ \\
\bottomrule
\end{tabular}
\end{table}
\end{Example}

Since the algorithm performs a double loop over all transaction blocks, the complexity of {\sc ComputeBlockSizes} equals $O(\abs{\trans_{\coll}}^{2})$. Note that topologically sorting the blocks (line~\ref{algo:sortblocks}) takes $O(\abs{\trans_{\coll}} \log \abs{\trans_{\coll}}) \leq O(\abs{\trans_{\coll}} k)$, however, we can also simply ensure that the blocks are topologically sorted by construction.

In this work, we substantially improve upon the {\sc ComputeBlockSizes} algorithm, by using a generalized version of the Quick Inclusion-Exclusion ({\sc qie}) algorithm, introduced by \citet{calders:05:quick}. The new algorithm presented here, called {\sc QieBlockSizes}, has a lower complexity than {\sc ComputeBlockSizes}.
The idea behind Quick Inclusion-Exclusion is to reuse intermediate results to reduce the number of subtractions.
The standard {\sc qie} algorithm computes the supports of all generalized itemsets based on some given itemset of size $k$ (a generalized itemset is an itemset containing both positive and negative items, e.g., $a\overline{b}$ means $a$ and not $b$), using the supports of all of its (positive) subsets. For instance, from the supports of $\mathit{ab}$, $\mathit{a}$, $\mathit{b}$, and the empty set, we can infer the support of $\overline{\mathit{ab}}$: $\supp{\overline{ab}}=\supp{\emptyset}-\supp{a}-\supp{b}+\supp{ab}$. {\sc qie} therefore works on an array of size $2^{k}$, which allows an implementation of the algorithm to employ efficient array indexing using integers, and makes it easy to locate subsets using bit operations on the indices---where a positive item is represented by a 1 and a negative item by a 0. 

In our setting, we want to find the sizes of transaction blocks which correspond to subsets of $\coll$, starting from the cumulative sizes of said blocks. We can represent each block $T$ by a binary vector defined by the indicator functions $S_{X}$. However, an important difference with the {\sc qie} algorithm is that not every possible binary vector necessarily corresponds to a (nonempty) transaction block, i.e., it is possible that $\abs{\trans_{\coll}} < 2^{k}$. Clearly, if $\abs{\trans_{\coll}}\ll 2^{k}$, it would be inefficient to use an array of size $2^{k}$. Therefore, we must take this fact into account. Before we can discuss the algorithm itself, we first need to introduce the following definitions.

\begin{definition}
Given a collection of itemsets $\coll=\langle X_{1}, \ldots, X_{k}\rangle$ and an integer $j\in\{0,\ldots,k\}$, the \emph{$j$-prefix} of $\coll$ is defined as
\[
\coll_{j} = \{X_{1},\ldots,X_{j}\}\;.
\]
For a subcollection $\ifam{G}$ of $\coll$, we define the $\mathit{closure} : 2^{\coll}\rightarrow 2^{\coll}$ as
\[
	\closure{\ifam{G}} = \{X_{i} \in \coll \mid X_{i}\subseteq\bigcup_{X\in \ifam{G}}X \}\;.
\]
The \emph{$j$-closure} of $\ifam{G}$ is defined as
\begin{align*}
	\closure{\ifam{G}, j} &= \ifam{G} \cup \{X_{i} \in \coll \mid X_{i}\notin\coll_{j} \mathrm{\ and\ } X_{i}\subseteq\bigcup_{X\in \ifam{G}}X \}\\
	& = \ifam{G} \cup \left( \closure{\ifam{G}} \setminus \coll_{j} \right)
\end{align*}
\end{definition}

The following lemma states that there is a one-to-one mapping between the closed subsets $\ifam{G}$ of $\coll$, and the transaction blocks of $\trans_{\coll}$ .

\begin{lemma}
\label{lem:closure1}
Let $\ifam{G}$ be an itemset collection. Then $\ifam{G} = \closure{\ifam{G}}$ if and only if
there exists a block $T$ in $\trans_{\coll}$ such that $\ifam{G} = \sets{T; \coll}$.
\end{lemma}

\begin{proof}
Assume that $\ifam{G} = \closure{\ifam{G}}$. Let $U = \bigcup_{X \in \ifam{G}}
X$ and let $t \in \trans$ be such that $t_i = 1$, if $a_i \in U$, and $t_i = 0$ otherwise.
Let $T \in \trans_\coll$ be the block containing $t$. If $X \in \ifam{G}$, then $S_X(t) = 1$.
On the other hand, if $S_X(t) = 1$, then $X \subseteq U$ and consequently $X \in \closure{\ifam{G}} = \ifam{G}$.
Hence, $\ifam{G} = \sets{T; \coll}$.

Assume now that there is a $T$ such that $\ifam{G} = \sets{T; \coll}$, let $t \in
T$ and $U = \bigcup_{X \in \ifam{G}} X$. It follows that $S_U(t) = 1$.  Let $X
\in \closure{\ifam{G}}$, then $X \subseteq U$ and $S_X(t) = 1$.  Hence, $X \in
\sets{T; \coll} = \ifam{G}$. Since $\ifam{G} \subseteq \closure{\ifam{G}}$, the lemma follows.
\end{proof}

Using the above lemma, we can introduce the following function, which maps subsets of $\coll$ to their corresponding blocks.
\begin{definition}
For a subset $\ifam{G}$ of a collection of itemsets $\coll$, we define
\[
	\block{\ifam{G}} = 
	\begin{cases}
		T \in \trans_{\coll} \mathrm{\ s.t.\ } \sets{T;\ifam{C}}=\ifam{G} & \mathrm{if\ } \closure{\ifam{G}}=\ifam{G}\;, \\
		\emptyset & \mathrm{otherwise}\;.\\
	\end{cases}
\]
That is, if $\ifam{G}$ is closed, the $\mathit{block}$ function simply maps it to the corresponding block in $\trans_{\coll}$. If $\ifam{G}$ is not closed, it is mapped to the empty transaction block.
Note that $\block{\sets{T;\coll}}=T$ for all $T\in\trans_{\coll}$. 
\end{definition}

\begin{algorithm}[t]
\caption{{\sc QieBlockSizes}($\coll$)}
\label{algo:qie}
\DontPrintSemicolon
\SetKwInOut{Input}{input}\SetKwInOut{Output}{output}

\Input{itemset collection $\coll=\{X_{1}, \ldots, X_{k}\}$}
\Output{block sizes $e(T)$ for each $T$ in $\trans_{\coll}$}

\For{$T$ in $\mathcal{T_{\coll}}$} {
	$I \define \bigcup \{X \mid X \in \sets{T ; \coll}$\}\;
	$\cum{T} \define 2^{N - \abs{I}}$\; \nllabel{algo:qie:cumuprob}
}

\For {$i=1,\ldots,k$} {\nllabel{algo:qie:iterk}
	\ForEach{$T \mathrm{\ in\ } \ifam{T}_{\coll}$}{
		$\ifam{G} \leftarrow \sets{T;\coll}$\;
		\If{$X_{i} \notin \ifam{G}$}{
			$\ifam{G}' \leftarrow \closure{\ifam{G}\cup \{X_{i}\}, i - 1}$\;
			$T' \leftarrow \block{\ifam{G}'}$\;
			\If{$T' \neq \emptyset$}{ \nllabel{algo:qie:findblock}
				$e(T) \leftarrow e(T) -e(T')$\;
			}
		}
	}
}
\end{algorithm}

{\sc QieBlockSizes} is given as Algorithm~\ref{algo:qie}. As before, we first compute the cumulative size of every block $T$ in $\trans_{\coll}$ (line~\ref{algo:qie:cumuprob}). Then, for each itemset $X_{i}$ (line~\ref{algo:qie:iterk}), the algorithm subtracts from each block $T$ for which $X_{i}\notin \ifam{G}=\sets{T;\coll}$, the current size of the block $T'$ corresponding to $\ifam{G}'=\closure{\ifam{G} \cup \{X_{i}\}, i-1}$ if $T'$ exists in $\trans_{\coll}$, i.e., the size of $T'=\block{\ifam{G}'}$. 

The following theorem states that {\sc QieBlockSizes} correctly computes the sizes of all blocks of transactions in $\trans_{\coll}$. For the proof, please refer to the Appendix.

\begin{theorem}
\label{theorem:itscorrect}
Given a collection of itemsets $\coll=\{X_{1},\ldots,X_{k}\}$, let $\trans_{\coll}$ be the corresponding partition with respect to $\coll$.
The algorithm {\sc QieBlockSizes} correctly computes the block sizes $e(T)$ for $T\in \trans_{\coll}$.
\end{theorem}

\begin{Example}
Let us apply {\sc QieBlockSizes} to our running example. Recall that $\coll= \{\mathit{abc}, \mathit{cd}, \mathit{def} \}$ (see Table~\ref{fig:blocks}). For brevity, we restrict ourselves to the first three blocks. In step 1, the first three blocks remain unaffected, since they all contain $X_{1}$. In step 2, only the third block does not contain $X_{2}$; we subtract the second block from it, to obtain $32-16=16$. In step 3, we subtract the first block from the second block, and get $16-4=12$. From the third block we do not have to subtract anything, since the $3$-closure of the corresponding block does not appear in $\trans_{\coll}$. We have thus calculated the sizes of the first three blocks using two subtractions, rather than three, as was previously required in Example~\ref{example:regularsubtract}.
\end{Example}

Finally, we can significantly optimize the algorithm as follows. Assume that we can
divide $\coll$ into two disjoint groups $\coll_1$ and $\coll_2$, such that if
$X_1 \in \coll_1$ and $X_2 \in \coll_2$, then $X_1 \cap X_2 = \emptyset$.
Let $B = \bigcup \coll_1$ be the set of items occurring in $\coll_{1}$.
Theorem~\ref{thr:exponential} implies that $\pemp(A) = \pemp(B)\pemp(A \setminus B)$. 
In other words, the maximum entropy distribution can be factorized into two
\emph{independent} distributions, namely $\pemp(B)$ and $\pemp(A \setminus B)$, and more
importantly, the factor $\pemp(B)$ depends only on $\coll_1$. Consequently,
if we wish to compute the probability $\pemp(X = 1)$ such that $X \subset B$, we can
ignore all variables outside $B$ and all itemsets outside $\coll_1$.
The number of computations to be performed by \textsc{QieBlockSizes}
can now be greatly reduced, since in the case of independence it holds that
$\abs{\trans_{\coll}} = \abs{\trans_{\coll_{1}}} \times \abs{\trans_{\coll_{2}}}$,
and we can simply compute the block sizes for $\trans_{\coll_{1}}$
and $\trans_{\coll_{2}}$ separately. 
Naturally, this decomposition can also be applied when there are more than two disjoint groups of itemsets.

Moreover, in order to \emph{guarantee} that we can apply the above separation, we could reduce the solution space slightly by imposing a limit on the number of items (or itemsets) per group, such that the number of blocks for each group remains small. 
Alternatively, we could first partition the items of the dataset into smaller, approximately independent groups, and subsequently apply the algorithm for each group separately---for which the approach of~\citet{mampaey:10:summarising} of  identifying the optimal partitioning by {\sc mdl} would be a logical choice.

\subsection{Querying the Model}
\label{sec:modelquery}
We have seen how we can efficiently query the probability of an itemset $X \in \coll$ when given the maximum entropy distribution $\pemp_{\coll}$. In order to compute the probability of an arbitrary itemset $Y$ that is not a member of $\coll$, we do the following.
We first set $\coll' = \coll \cup \{Y\}$ and compute the block probabilities $\exact{T'}$ for all $T'$ in $\trans_{\coll'}$ by calling \textsc{QieBlockSizes}.
Then, we can simply use the parameters of $\pemp_{\coll}$ to compute $\pemp_{\coll}(Y=1)$ as follows,
\[
\pemp_{\coll}(Y = 1) = \sum_{T \in \trans_{\coll'} \atop S_Y(T) = 1} e(T) \prod_{X \in \coll} u_{X}^{S_{X}(T)}\;.
\]
Thus, to obtain the probability of an itemset, it suffices to compute the block probabilities in $\trans_{\coll'}$, 
for which we know that $\abs{\trans_{\coll'}} \leq 2 \abs{\trans_{\coll}}$.

\subsection{Computational Complexity}
Let us analyze the complexity of the \textsc{IterativeScaling} algorithm.  
To this end, we define $\partsize{\coll} = \abs{\trans_{\coll}}$ as the number of blocks in $\trans_{\coll}$. 
The computational complexity of \textsc{QieBlockSizes} is 
\[
O\left(k\cdot \partsize{\coll} \log\partsize{\coll}\right)\;,
\]
for a given collection $\coll$, with $\abs{\coll}=k$. The logarithmic factor comes from looking for the block $T'$ on line~\ref{algo:qie:findblock}.
Note that $\partsize{\coll} \leq 2^k$, and hence $\log \partsize{\coll} \leq k$.
Assume now that we can partition $\coll$ into $L$ disjoint parts $\coll=\coll_{1} \cup \cdots \cup \coll_{L}$,
such that if $X\in \coll_{i}$ and $Y\in \coll_j$ then $X\cap Y=\emptyset$.
As mentioned in Section~\ref{sec:modelcompute}, we can now simply compute $L$ independent distributions at a lower total cost.
Denoting $B_{i} = \bigcup_{X\in \coll_{i}} X$, it holds that
$\partsize{\coll_{i}} \leq \min\fpr{2^\abs{\coll_{i}}, 2^\abs{B_{i}}}$. 
If $\coll_{i}$ cannot be partitioned further, this usually means that either $\abs{\coll_{i}}$ is small,
or the itemsets in $\coll_{i}$ overlap a lot and $\partsize{\coll_{i}} \ll 2^{\abs{\coll_{i}}}$.
The total execution time of \textsc{IterativeScaling} is 
\[
O\left(K \sum_{i=1}^{L} \abs{C_{i}} \partsize{\coll_{i}}\log \partsize{\coll_{i}}\right)\;,
\]
where $K$ is the number of iterations, which is usually low.
The complexity of estimating the frequency of an itemset requires running 
\textsc{QieBlockSizes} once and, hence equals 
\[
O\left(\sum_{i=1}^{L} \abs{C_{i}}\partsize{\coll_{i}}\log\partsize{\coll_{i}}\right)\;.
\]

\subsection{Including Background Knowledge into the Model}
\label{sec:background}

Typically, when analyzing data we have some basic background knowledge about the data. For instance, we may already know the individual frequencies of the items, i.e., the \emph{column margins}. These margins supply some basic information about the data, for instance whether {\em tomatoes} are sold often or not in a supermarket database. These individual frequencies are intuitive and easy to calculate, yet already provide information on whether some combinations of items are more or less likely to occur frequently. For this reason, many existing techniques use the independence model as a basis to discover interesting patterns, \citep[e.g.,][]{brin:97:beyond,aggarwal:98:new}. Another form of background information that is often used are the \emph{row margins} of the data, that is, the probabilities that a transaction contains a certain number of items, e.g., \citep{gionis:07:assessing,hanhijarvi:09:tell,konto:10:sdm,tatti:10:bgrank}. If we know that most transactions are rather small, large itemsets are likely to have low frequencies.

Clearly, when we analyze data we want to incorporate this background knowledge, since otherwise we would simply rediscover it. If we do include it in our analysis, we discover itemsets that are interesting \emph{with respect to} what we already know. Therefore, we extend the {\sc bic} and {\sc mdl} quality scores of Definition~\ref{def:mdlscore} to incorporate background knowledge, say, $\bg$. Although in this section we focus on row and column margins as forms of background knowledge, many other patterns or count statistics that can be expressed as linear constraints on transactions could be used, for instance, a set of association rules and their confidences. Therefore, we intentionally omit the specification of $\bg$.

\begin{definition}
Given a dataset $\db$ and some background knowledge $\bg$, we define the {\sc bic} score for a collection of itemsets $\coll=\enset{X_{1}}{X_{k}}$, with respect to $\bg$ as
\begin{equation}
	\bicscore{\coll, \db ; \bg} = -\log \pemp_{\bg, \coll}\fpr{\db} +k/2 \log\abs{D}\;.
\end{equation}
Similarly, we define the {\sc mdl} score of $\coll$ with respect to $\bg$ as
\begin{equation}
	\mdlscore{\coll, \db ; \bg} = -\log \pemp_{\bg, \coll}\fpr{\db} + l_{1} k + l_{2} x + 1\;,
\end{equation}
where $\pemp_{\bg, \coll}$ is the maximum entropy distribution satisfying the background knowledge $\bg$ and $\pemp_{\bg, \coll}(X=1)=\freq{X}$ for all $X \in \coll$, and $l_{1}$ and $l_{2}$ are the same as in Definition~\ref{def:mdlscore}.
\end{definition}

Note that while the background knowledge is included in the log-likelihood term of $\score{\coll,\db ; \bg}$ (where $s$ denotes either {\sc bic} or {\sc mdl}), it is not included in the penalty term. We choose not to do so because we will assume that our background knowledge is consistent with the data and invariable. We could alternatively define
\[
	\score{\coll, \db, \bg} = \score{\coll, \db ; \bg} + L(\bg)\;,
\]
where $L(\bg)$ is some term which penalizes $\bg$, however, since this term would be equal for all $\coll$, for simplicity it might as well be omitted. If we were to compare different models that use different background knowledge, though, it must be included.

In this section, we show how to include row and column margins as background knowledge into our algorithm without, however, blowing up its computational complexity. If we were, for instance, to naively add all singleton itemsets $\ifam{I}$ and their frequencies to an itemset collection $\coll$, the number of transaction blocks in the corresponding partition would become $\abs{\trans_{\coll \cup \ifam{I}}} = \abs{\trans} = 2^{N}$\!, by which we would be back at square one. Therefore, we will consider the row and column margins separately from $\coll$ in our computations.

First, let us consider using only column margins, viz., item frequencies. With these, we build an independence model, while with $\coll$ we partition the transactions $\trans$ as above; then we simply combine the two to obtain the maximum entropy distribution. 
As before, the maximum entropy model has an exponential form:
\begin{equation}
	\pemp_{\coll'}(A=t) = u_{0} \prod_{X \in \coll} u_{X}^{S_{X}(t)} \prod_{i \in \ifam{I}} v_{i}^{S_{i}(t)}\;.
\end{equation}
The second part of the product defines an independence distribution
\begin{equation}
	v(A=t) = v_{0}\prod_{i}v_{i}^{S_{i}(t)}\;,
\end{equation}
where $v_{0}$ is a normalization factor. It is not difficult to see that $v(a_{i}=1)=v_{i}/(1+v_{i})$, for all for $a_{i} \in A$. It should be noted that while $\pemp(a_{i}=1)=\freq{a_{i}}$, in general it does not necessarily hold that $v(a_{i}=1)=\freq{a_{i}}$. Now we can write
\begin{equation}
	\pemp_{\coll'}(A\in T)=v(A\in T) \frac{u_{0}}{v_{0}} \prod_{X \in \coll} u_{X}^{S_{X}(T)}\;.
\end{equation}
Thus, we simply need to compute $v(A\in T)$ for each block $T$, which is computed very similarly to $e(T)$, using \textsc{QieBlockSizes}. Note that $e(A=t)$ is in fact nothing more than a uniform distribution over $\trans$, multiplied by $2^{N}$\!. To compute $v(A \in T)$, we simply initiate the algorithm with the cumulative sizes of the blocks with respect to $v$, which are equal to 
\[c(A \in T)=v_{0} \prod_{i\in I} v_{i}\;,\]
 where $I=\bigcup\sets{T;\coll}$.
Hence, we can include the item frequencies at a negligible additional cost. To update the $v_{i}$ parameters, we must query the probability of a single item. We can achieve this by simply adding the corresponding singleton to $\coll$, in exactly the same way as described in Section~\ref{sec:modelquery}.

Next, we also include row margins in the background information.
Let us define the indicator functions $S^{j}(t):\trans \rightarrow \{0,1\}$ for $j \in \{0, \ldots, N\}$ such that $S^{j}(t)=1$ if and only if the number of ones in $t$, denoted as $\abs{t}$, is equal to $j$. Further, for any distribution $p$ on $A$, let us write $p(\abs{A}=j)$ to indicate the probability that a transaction contains $j$ items. Again, the maximum entropy distribution has an exponential form,
\begin{equation}
	\pemp_{\bg,\coll}(A=t) = u_{0} \prod_{X\in\coll} u_{X}^{S_{X}(t)} \prod_{i \in \ifam{I}} v_{i}^{S_{i}(t)}\prod_{j=0}^{N} w_{j}^{S^{j}(t)}\;.
\end{equation}
The row and column margins define a distribution
\begin{equation}
	w(A=t) = v_{0} \prod_{i \in \ifam{I}} v_{i}^{S_{i}(t)}\prod_{j=0}^{N} w_{j}^{S^{j}(t)}\;,
\end{equation}
where $v_{0}$ is a normalization factor.
Now, for the probabilities $\pemp(A\in T)$, we have
\begin{align*}
	\pemp(A\in T, \abs{A}=j) &= w(A \in T, \abs{A}=j) \frac{u_{0}}{v_{0}} \prod_{X \in \coll} u_{X}^{S_{X}(T)}\\
	&=w_{j} v(A\in T, \abs{A}=j) \frac{u_{0}}{v_{0}} \prod_{X \in \coll} u_{X}^{S_{X}(T)}
\end{align*}
for $j=0,\ldots,N$, and we marginalize over $j$ to obtain 
\begin{align*}
	\pemp(A\in T) &= \sum_{j=0}^{N}\pemp(A\in T, \abs{A}=j) \\
			& = \frac{u_{0}}{v_{0}} \prod_{X \in \coll} u_{X}^{S_{X}(T)} \sum_{j=0}^{N} w_{j} v(A\in T, \abs{A}=j)
\end{align*}
As above, we compute the probabilities $v(A\in T, \abs{A}=j)$ using {\sc QieBlockSizes}. 
Let $I=\bigcup\{X \mid X \in \sets{\coll;T}\}$, then the corresponding cumulative probability becomes 
\[
	c(A\in T, \abs{A}=j) = v_{0}\prod_{i\in I} v_{i} \cdot v(\abs{A}=j \mid I=1)
\]

Computing the probabilities $v(\abs{A}=j)$, and similarly $v(\abs{A}=j \mid I=1)$, can be done from scratch in $O(N^{2})$ time and $O(N)$ space, using the following recurrence relation,
\[
	v(\abs{A_{i}} = j) = v(a_{i})\cdot v(\abs{A_{i-1}} = j-1) + (1-v(a_{i})) \cdot v(\abs{A_{i-1}} = j)\;,
\]
where $A_{i} = \enset{a_{1}}{a_{i}}$. Starting from $A_{0}=\emptyset$, the {\sc ComputeSizeProbabilities} algorithm adds each item $a_{i}$ until we have computed all probabilities $v(\abs{A_{N}}=j)$ where $A_{N}=A$; see Algorithm~\ref{algo:computeprobs}. The time complexity can be reduced to $O(N)$, by applying the updates that {\sc IterativeScaling} performs on $v(a_{i})$, to the probabilities $v(\abs{A}=j)$ as well, this is done by {\sc UpdateSizeProbabilities} in Algorithm~\ref{algo:updateprobs}. The algorithm first removes the item, and then re-adds it with the updates probability. For further details we refer to \citet{tatti:10:bgrank}. 
Computing item frequencies is done by adding singletons to $\coll$.
To compute the row margin probabilities $\pemp(\abs{A}=j)$, we simply marginalize over $\trans_{\coll}$,
\[
	\pemp(\abs{A}=j) = \sum_{T \in \trans_{\coll}} \pemp(A\in T, \abs{A}=j)\;.
\]
Hence, including the row margins increases the time and space complexity of model computation and inference by a factor of $N$.

\begin{algorithm}[ht!]
\caption{\textsc{ComputeSizeProbabilities}($v$)}
\label{algo:computeprobs}
\DontPrintSemicolon
\SetKwInOut{Input}{input}\SetKwInOut{Output}{output}

\Input {independence distribution $v$ over $A$, with probabilities $v_{i}=v(a_{i}=1)$ for $i=1, \ldots, N$}
\Output {probabilities $g_{j} = v(\abs{A}=j)$ for $j=0, \ldots, N$}

$g_{0} \define 1$\;
\For{$j=1,\ldots,N$} {
	$g_{j} \define 0$\;
}

\For{$i=1,\ldots,N$} {
	\For {$j=i, \ldots, 1$} {
		$g_{j} \define v_{i} \cdot g_{j-1} + (1-v_{i}) \cdot g_{j}$ \;
	}
	$g_{0} \define (1-v_{i}) \cdot g_{0}$ \;
}
\Return $g$ \;
\end{algorithm}

\begin{algorithm}[ht!]
\caption{\textsc{UpdateSizeProbabilities}($v$, $g$, $a_{i}$, $x$)}
\label{algo:updateprobs}
\DontPrintSemicolon
\SetKwInOut{Input}{input}\SetKwInOut{Output}{output}

\Input {probabilities $g_{j}=v(\abs{A}=j)$ for independence distribution $v$, with parameters $v_{i}=v(a_{i}=1)$, updated probability $v_{i}'$ for item $a_{i}$}
\Output {updated probabilities $g_{j} = v(\abs{A}=j)$ for $j=0, \ldots, N$}

$g_{0} \define g_{0} / (1-v_{i})$ \;

\For{$j=1,\ldots,N$} {
	$g_{j} \define \left(g_{j} - v_{i} g_{j-1}\right) / (1-v_{i})$ \;
}
update $v$ such that $v(a_{i}=1) = v_{i}'$ \;
\For {$j=N, \ldots, 1$} {
	$g_{j} \define v_{i}' \cdot g_{j-1} + (1-v_{i}') \cdot g_{j}$ \;
}
$g_{0} \define (1-v_{i}') \cdot g_{0}$ \;

\Return $g$ \;
\end{algorithm}

%% file: problem.tex
\section{Problem Statements}
\label{sec:problem}

In this section we identify four different problems that we intend to solve using the theory introduced above. 
We assume some given set $\bg$ that represents our background knowledge, 
e.g., the individual item frequencies, some arbitrary collection of itemsets, or simply the empty set.
We start simple, with a size constraint $k$ and a collection $\ifam{F}$ of potentially interesting itemsets to choose from, for instance, frequent itemsets, closed itemsets, itemsets of a certain size, etc.

\begin{prblm}[Most Informative $k$-Subset]\label{prb:miks}
Given a dataset $\db$, a set $\bg$ that represents our background knowledge, an integer $k$, and a collection of potentially interesting itemsets $\ifam{F}$, find the subset $\coll \subseteq \ifam{F}$ with $\abs{\coll} \leq k$ such that $\score{\coll, \db ; \bg}$ is minimal.
\end{prblm}
 
Note that if we choose $k = 1$, this problem reduces to `Find the Most Interesting Itemset in $\ifam{F}$', which means simply scoring $\coll=\{X\}$ with respect to $\bg$ for each set $X \in \ifam{F}$, and selecting the best one. Further, these scores provide a ranking of 
the itemsets $X \in \ifam{F}$ with regard to what we already know, that is, $\bg$.

Now, if we do not want to set $k$ ourselves, we can rely on either {\sc bic} or {\sc mdl} to identify the best-fitting, least-redundant model, a problem we state as the following.

\begin{prblm}[Most Informative Subset]\label{prb:mis}
Given a dataset $\db$, a set $\bg$ that represents our background knowledge, and a collection of potentially interesting itemsets $\ifam{F}$, find the subset $\coll \subseteq \ifam{F}$ such that $\score{\coll, \db ; \bg}$ is minimal.
\end{prblm}

When we do not want to constrain ourselves to a particular itemset collection $\ifam{F}$, we simply use all itemsets. Problem~\ref{prb:miks} then generalizes to the following.
\begin{prblm}[$k$ Most Informative Itemsets]\label{prb:miki}
Given a dataset $\db$, an integer $k$, and a set $\bg$ that represents our background knowledge, find the collection of itemsets $\coll$, with $\abs{\coll} \leq k$, such that $\score{\coll, \db ; \bg}$ is minimal.
\end{prblm}

Similarly, and most generally, we can simply consider finding the best collection of itemsets altogether.

\begin{prblm}[Most Informative Itemsets]\label{prb:mii}
Given a dataset $\db$ and a set $\bg$ that represents our background knowledge, find the collection of itemsets $\coll$ such that $\score{\coll, \db ; \bg}$ is minimal.
\end{prblm}

Note that these problem statements do not require $\ifam{F}$ to be explicitly available beforehand (let alone the complete set of itemsets), i.e., it does not have to be mined or materialized in advance (we postpone the details of this to Section~\ref{sec:looking}). 

Next, we discuss how we can efficiently mine sets of itemsets to solve the above problems.

%% file: algorithm.tex
\section{Mining Informative Succinct Summaries}
\label{sec:algorithm}

In Section~\ref{sec:theory} we described how to compute the maximum entropy model 
and its {\sc bic} or {\sc mdl} quality score \emph{given} a set of itemsets. 
Finding the \emph{optimal} collection as stated in Section~\ref{sec:problem}, however, 
is clearly nontrivial. The size of the search space is
\[
	\sum_{j=0}^{k}{\abs{\ifam{F}} \choose j} \leq 2^{\abs{\ifam{F}}}\;.
\]
If we do not restrict the candidate itemsets, then the
number of all (non-singleton) itemsets is $\abs{\ifam{F}} = 2^N - N - 1$.
Moreover, our quality scores are not monotonic, nor is there to our knowledge some
easily exploitable structure, which prevents us from straightforwardly exploring the search space.

Therefore, we resort to using a heuristic, greedy approach. 
Starting with a set of background knowledge---for instance the column margins---we 
incrementally construct our summary by iteratively adding the 
itemset that reduces the score function the most. The algorithm stops either when 
$k$ interesting itemsets are found, or when the score 
no longer decreases. The pseudo-code for our {\sc mtv} algorithm, 
which mines Maximally informaTiVe itemset summaries, is given in Algorithm~\ref{alg:mtv}.

Due to its incremental nature, we note that we can apply an optimization to the algorithm.
When we call \textsc{IterativeScaling} on line~\ref{algo:bestiterscale}, 
rather than computing $\pemp$ from scratch, we can initialize the algorithm with the parameters of the previous $\pemp$ (line~\ref{algo:inititerscale} of Algorithm~\ref{algo:iterscale}), instead of with the uniform distribution. 
In doing so, the \textsc{Iterative\-Scaling} procedure converges faster.
Further, we can also reuse part of the computations from {\sc QieBlockSizes}.

\begin{algorithm}[t]
\DontPrintSemicolon
\Input {binary dataset $\db$, background knowledge $\bg$, integer $k\leq\infty$}
\Output {itemset collection $\coll$}

$\ifam{I} \define $ items in $\db$\;
\While {$\score{\coll, \db;\bg}$ decreases \AND $\abs{\coll} < k$} {
	$X \define \textsc{FindMostInformativeItemset}(\emptyset, \ifam{I}, \emptyset)$\;
	$\coll \define \coll \cup \set{X}$\;
	$\pemp_{\bg,\coll} \define \textsc{IterativeScaling}(\coll)$\; \nllabel{algo:bestiterscale}
	compute $\score{\coll, \db; \bg}$\;
}
\Return $\coll$\;
\caption{\textsc{mtv}($\db$, $\bg$, $k$)}
\label{alg:mtv}
\end{algorithm}

\subsection{A Heuristic for Scoring Itemsets}

Finding the most informative itemset to add to the current collection is practically
infeasible, since it involves solving the maximum entropy model for each
and every candidate. This remains infeasible even if we restrict the search 
space (for example, using only frequent itemsets). Therefore, instead of 
selecting the candidate that optimizes the {\sc bic} or {\sc mdl} score directly, we select the 
candidate that maximizes a heuristic which expresses the divergence 
between its frequency and its estimate. 
To derive and motivate this heuristic we first present the following theorem.

\begin{theorem}
\label{th:argmax}
Given an itemset collection $\coll$, a dataset $\db$, and a candidate collection of itemsets $\ifam{F}$. Let $s$ denote either {\sc bic} or {\sc mdl}. It holds that
\begin{align*}
\argmin_{X\in\ifam{F}} \score{\coll\cup \{X\}} &=\argmax_{X\in\ifam{F}} \kl{\pemp_{\coll\cup \{X\}}}{\pemp_{\coll}} - r(X) \\
 &=\argmin_{X\in\ifam{F}} \kl{q_{\db}}{\pemp_{\coll\cup\{X\}}} + r(X)
\end{align*}
where 
\[
	r(X) = 	
		\begin{cases}
			0 &\textrm{if}\quad s=\bic\\
			\abs{X}\log N/\abs{D} &\textrm{if}\quad s=\mdl
		\end{cases}
\]
\end{theorem}
\begin{proof}
Let us write $\coll' = \coll \cup \set{X}$.
Corollary~\ref{cor:likelihood} states that 
$
	- \log \pemp_{\coll'}(\db) = \abs{\db}\ent{\pemp_{\coll'}}.
$
In addition, we can show with a straightforward calculation that
\[
	\kl{\pemp_{\coll'}}{\pemp_{\coll}} = \ent{\pemp_{\coll}} - \ent{\pemp_{\coll'}}\;.
\]
For {\sc bic} the difference in the penalty terms of $\score{\coll}$ and $\score{\coll'}$ is equal to $\nicefrac{1}{2}\log\abs{D}$, which is identical for all itemsets $X$, and hence may be eliminated from the $\argmax$.
For {\sc mdl}, the difference in penalty terms can similarly be reduced to $\abs{X}\log N$. The second equality follows similarly.
\end{proof}

Thus, we search for the itemset $X$ for which the new distribution diverges maximally from the previous one, 
or equivalently, brings us as close to the empirical distribution as possible---taking into account the penalty term for {\sc mdl}. Note that for {\sc bic}, since $r(X)=0$, the algorithm simply tries to maximize the likelihood of the model, and the penalty term functions as a stopping criterion; the algorithm terminates when the increase in likelihood (i.e., decrease of the negative log-likelihood) is not sufficient to counter the increase of the penalty. When we use {\sc mdl}, on the other hand, $r(X)$ represents a part of the penalty term, and hence this guides the algorithm in its search.

The heuristic we employ uses an approximation of the above $\mathit{KL}$ divergence,
and is in fact a simpler $\mathit{KL}$ divergence itself. In the expression 
\begin{equation}
\label{equation:kl}
	\kl{\pemp_{\coll'}}{\pemp_{\coll}}=\sum_{t\in\trans} \pemp_{\coll'}(A=t) \log \frac{\pemp_{\coll'}(A=t)}{\pemp_{\coll}(A=t)}\;
\end{equation}
we merge the terms containing $X$ into one term, and the terms not containing $X$ into another term.
To differentiate between these two divergences, let us write the function $\mathit{kl} : [0,1]\times[0,1] \rightarrow \mathbb{R}^{+}$ as follows,
\begin{equation}
	\mathit{kl}(x, y) = x \log \frac{x}{y} + (1 - x) \log \frac{1 - x}{1 - y}\;,
\end{equation}
then we approximate Eq.~\ref{equation:kl} by $\mathit{kl}\left(\freq{X}, \pemp_{\coll}(X = 1)\right)$.
We will write the latter simply as $\mathit{kl}(X)$ when $\mathit{fr}$ and $\pemp$ are clear from the context. To compute this heuristic, we only need the frequency of $X$, and its estimate according to the current $\pemp$ distribution. This gives us a measure of the divergence between $\freq{X}$ and $\pemp_{\coll}(X = 1)$, i.e., its surprisingness given the current model.

The following theorem shows the relation between $\mathit{KL}$ and $\mathit{kl}$.
\begin{theorem}
For an itemset collection $\coll$ and an itemset $X$, it holds that 
\[
0 \leq \mathit{kl}(X) \leq \kl{\pemp_{\coll\cup \{X\}}}{\pemp_{\coll}}\,.
\] 
Moreover, $\mathit{kl}(X)=0$ if and only if $\kl{\pemp_{\coll\cup \{X\}}}{\pemp_{\coll}}=0$, i.e., when $\freq{X}=\pemp_{\coll}(X = 1)$.
\end{theorem}
\begin{proof}
Both inequalities follow directly from the log-sum inequality, which states that for any nonnegative numbers $a_{i},b_{i}$, with $a=\sum_{i} a_{i}$ and $b=\sum_{i} b_{i}$, it holds that
\[
	\sum_{i} a_i\log\frac{a_i}{b_i}\geq a\log\frac{a}{b}\;.
\]
For equality to zero, we have $\mathit{kl}(\freq{X}, \pemp_{\coll}(X=1))=0$ if and only if $\freq{X}=\pemp_{\coll}(X=1)$. In this case it holds that $\pemp_{\coll'}=\pemp_{\coll}$ which is true if and only if $\kl{\pemp_{\coll'}}{\pemp_{\coll}}=0$.
\end{proof}

Using Theorem~\ref{th:argmax}, the heuristic function we employ is defined as
\begin{equation}
	h(X) = \mathit{kl}(\freq{X}, \pemp_{\coll}(X=1)) - r(X)
\end{equation}
and we will make use of the following assumption: 
\[
	\argmin_{X\in\ifam{F}} \score{\coll \cup \{X\}} = \argmax_{X\in\ifam{F}} h(X)\;.
\]

Note that $h$ has an elegant interpretation: it is equal to 
the Kullback-Leibler divergence 
after exactly one step in the {\sc IterativeScaling} algorithm---when initializing $\pemp$ with the parameters from the previous model, as discussed above.
Since the $\mathit{KL}$ divergence increases monotonically during the Iterative Scaling procedure, if the total divergence is large, then we expect to already see this in the first step of the procedure.

\subsection{Searching for the Most Informative Itemset}
\label{sec:looking}

\begin{algorithm}[t]
\DontPrintSemicolon
\Input {itemset $X$, remaining items $Y$, currently best set $Z$}
\Output {itemset between $X$ and $XY$ maximizing $h$, or $Z$}

compute $\freq{X}$ and $\pemp(X)$\; \nllabel{algo:freqfly}
\If {$h(X) = \mathit{kl}(\freq{X}, \pemp(X))-r(X) > h(Z)$}  {
	$Z \define X$\;
}
compute $\freq{XY}$ and $\pemp(XY)$\;
$\mathit{b}\define\max\{\mathit{kl}(\freq{X},\pemp(XY)), \mathit{kl}(\freq{XY},\pemp(X))\}-r(X)$\;
\If {$\mathit{b} > h(Z)$} {
	\For {$y \in Y$} {
		$Y \define Y \setminus \{y\}$\;
		$Z \define$ \textsc{FindMostInformativeItemset}($X \cup \{y\}$, $Y$, $Z$)\;
	}
}
\Return $Z$\;
\caption{\textsc{FindMostInformativeItemset}($X$, $Y$, $Z$)
}
\label{algo:bestitemset}
\end{algorithm}

To find the itemset maximizing $h(X)$, 
we take a depth-first branch-and-bound approach. 
We exploit the fact that $\mathit{kl}$ is a convex function, and
employ the bound introduced by \citet{nijssen:09:correlated} 
to prune large parts of the search space as follows.
Say that for a candidate itemset $X$ in the search space,
its maximal possible extension in the branch below it is $X\cup Y$ (denoted $XY$), then for any itemset $W$ such that 
$X \subseteq W \subseteq XY$ it holds that
\begin{equation}
	h(W) = \mathit{kl}(W) - r(W) \leq \max \Bigl\{ \mathit{kl}\bigl( \freq{X}\!, \pemp(XY)\bigr), \mathit{kl}\bigl( \freq{XY}\!, \pemp(X)\bigr)\Bigr\} - r(X)\,.
\end{equation}
If this upper bound is lower than the best value of the heuristic seen so far, we know that no (local) extension $W$ of $X$
can ever become the best itemset with respect to the heuristic, and therefore we can safely prune the branch of the search space below $X$. 
The \textsc{FindMostInformativeItemset} algorithm is given in Algorithm~\ref{algo:bestitemset}.

An advantage of this approach is that we do not need to collect the frequencies of all candidate itemsets beforehand.
Instead, we just compute them on the fly when we need them (line~\ref{algo:freqfly}).
For instance, if we wish to pick itemsets from a collection $\ifam{F}$ of frequent itemsets for some minimum support
threshold, we can integrate the support counting with the depth-first traversal of the algorithm, rather than first 
mining and storing $\ifam{F}$ in its entirety.
Since for real datasets and nontrivial minimal support thresholds billions of frequent itemsets are easily discovered, this indubitably makes our approach more practical.

%% file: experiments.tex
\section{Experiments}
\label{sec:experiments}
In this section we experimentally evaluate our method and empirically validate the quality of the returned summaries.

\subsection{Setup}
We implemented our algorithm in C++, and provide the source code for research purposes.\!\footnote{http://www.adrem.ua.ac.be/implementations}
All experiments were executed on a 2.67GHz (six-core) Intel Xeon machine with 12GB of memory, running Linux. All reported timings are of the single-threaded implementation of our algorithm.

We evaluate our method on three synthetic datasets, as well as on eleven real datasets. Their basic characteristics are given in Table~\ref{tab:datasets}.

The \emph{Independent} data has independent items with random frequencies between 0.2 and 0.8. In the \emph{Markov} dataset each item is a noisy copy of the previous one, with a random copy probability between 0.5 and 0.8. 
The \emph{Mosaic} dataset is generated by randomly planting five itemsets of size 5 with random frequencies between 0.2 and 0.5, in a database with 1\% noise. 

The {\em Abstracts} dataset contains the abstracts of all accepted papers at the ICDM conference up to 2007, where all words have been stemmed and stop words have been removed~\cite{konto:10:sdm}.

The \emph{Accidents}, \emph{Kosarak}, \emph{Mushroom}, and \emph{Retail} datasets were obtained from the FIMI dataset repository~\citep{goethals:04:fimi,geurts:03:using,brijs:99:using}.

The \emph{Chess (kr--k)} and \emph{Plants} datasets were obtained from the UCI ML Repository~\citep{uciml},
the former was converted into binary form, by creating an item for each attribute-value pair. The latter contains a list of plants, and the U.S.\ and Canadian states where they occur.

The {\em DNA Amplification} data contains information on DNA copy number amplifications~\citep{myllykangas:06:dna}. Such copies activate oncogenes and are hallmarks of nearly all advanced tumors. 
Amplified genes represent attractive targets for therapy, diagnostics and prognostics. 
In this dataset items are genes, and transactions correspond to patients.

The {\em Lotto} dataset was obtained from the website of the Belgian National Lottery, and contains the results of all lottery draws between May 1983 and May 2011.\!\footnote{http://www.nationaleloterij.be} Each draw consist of seven numbers (six plus one bonus ball) out of a total of 42.

The {\em Mammals} presence data consists of presence records of European mammals within geographical areas 
of 50$\times$50 kilometers~\citep{mitchell-jones:99:atlas}.\!\footnote{The full version of the {\em Mammals} dataset is available for research purposes from the Societas Europaea Mammalogica at http://www.european-mammals.org} \citet{heikinheimo:07:mammals} analyzed the distribution of these mammals over these areas, and showed they form environmentally distinct and spatially coherent clusters.

The {\em MCADD} data was obtained from the Antwerp University Hospital. Medium-Chain Acyl-coenzyme A Dehydrogenase Deficiency (MCADD) \citep{baumgartner:05:modelling,vandenbulcke:11:data} is a deficiency newborn babies are screened for during a Guthrie test on a heel prick blood sample. The instances are represented by a set of 21 features: 12 different acylcarnitine concentrations measured by tandem mass spectrometry (TMS), together with 4 of their calculated ratios and 5 other biochemical parameters, each of which we discretized using \emph{k}-means clustering with a maximum of 10 clusters per feature.

\begin{table}[t!]
\caption{The synthetic and real datasets used in the experiments. Shown for each dataset are the number of items $\abs{A}$, the number of transactions $\abs{\db}$, the minimum support threshold for the set of candidate frequent itemsets $\ifam{F}$ and its size. 
}
\label{tab:datasets}
\centering
\begin{tabular*}{0.95\linewidth}{@{\extracolsep{\fill}}l rr rr }
\toprule
& \multicolumn{2}{c}{\textit{Data Properties}} & \multicolumn{2}{c}{\textit{Candidate Collection}}\\
\cmidrule{2-3}\cmidrule{4-5}
\textit{Dataset} & $\abs{A}$ & $\abs{\db}$ & $\mathit{minsup}$ & $\abs{\ifam{F}}$ \\ 
\midrule
Independent & 50 & 100\,000 & 5\,000 & 1\,055\,921 \\ 
Markov & 50 & 100\,000 & 5\,000 & 377\,011 \\ 
Mosaic & 50 & 100\,000 & 5\,000 & 101\,463 \\ [0.5em] 
Abstracts & 3\,933 & 859 & 10 & 75\,061 \\ 
Accidents & 468 & 340\,183 & 50\,000 & 2\,881\,487 \\ 
Chess (kr--k) & 58 & 28\,056 & 5 & 114\,148 \\ 
DNA Amplification & 391 & 4\,590 & 5 & 4.57$\cdot$10\textsuperscript{12} \\ 
Kosarak & 41\,270 & 990\,002 & 1\,000 & 711\,424 \\ 
Lotto & 42 & 2\,386 & 1 & 139\,127 \\ 
Mammals & 121 & 2\,183 & 200 & 93\,808\,244 \\ 
MCADD & 198 & 31\,924 & 50 & 1\,317\,234 \\ 
Mushroom & 119 & 8\,124 & 100 & 66\,076\,586 \\ 
Plants & 70 & 34\,781 & 2\,000 & 913\,440 \\ 
Retail & 16\,470 & 88\,162 & 10 & 189\,400 \\ 
\bottomrule
\end{tabular*}
\end{table}

The core of our method is parameter-free. That is, it will select itemsets from the complete space of possible itemsets. In practice, however, it may not always be feasible or desirable to consider \textit{all} itemsets. For dense or large datasets, for instance, we might want to explicitly exclude low-frequency itemsets, or itemsets containing many items.
General speaking, choosing a larger candidate space, yields a larger search space, and hence potentially better models. In our experiments we therefore consider collections of frequent itemsets $\ifam{F}$ mined at support thresholds as low as feasible. The actual thresholds and corresponding size of $\ifam{F}$ are depicted in Table~\ref{tab:datasets}. Note that although the minimum support threshold is used to limit the size of $\ifam{F}$, it can be seen as an additional parameter to the algorithm. 
To ensure efficient computation, we impose a maximum of 10 items per group, as described at the end of Section~\ref{sec:modelcompute}. Further, we terminate the algorithm if the runtime exceeds two hours.
For the sparse datasets with many transactions (the synthetic ones, \emph{Accidents}, \emph{Kosarak}, and \emph{Retail}), we mine and store the supports of the itemsets, rather than computing them on the fly. Caching the supports is faster in these cases, since for large datasets support counting is relatively expensive. The runtimes reported below include this additional step.

In all experiments, we set the background information $\bg$ to contain the column margins, i.e., we start from the independence model. For \emph{Chess (kr--k)} and \emph{Mushroom}, we also perform experiments which include the row margins, since their original form is categorical, and hence we know that each transaction has a fixed size. For the \emph{Lotto} data we additionally use \emph{only} the row margins, which implicitly assumes all numbers to be equiprobable, and uses the fact that each draw consist of seven numbers.

\subsection{A First Look at the Results}

In Tables~\ref{tab:bicscores}~and~\ref{tab:mdlscores} we present the scores and sizes of the discovered summaries, the time required to compute them, and for comparison we include the score of the background model, using {\sc bic} and {\sc mdl} respectively as quality score $s$.
For most of the datasets we consider, the algorithm identifies the optimum quite rapidly, i.e., within minutes. For these datasets, the number of discovered itemsets, $k$, is indicated in bold.
For three of the datasets, i.e., the dense \emph{MCADD} dataset, and the large \emph{Kosarak} and \emph{Retail} datasets, the algorithm did not find the optimum within two hours. 

For both {\sc bic} and {\sc mdl}, we note that the number of itemsets in the summaries is very small, and hence manual inspection of the results by an expert  is feasible. 
We see that for most datasets the score (and thus relatedly the negative log-likelihood) decreases a lot, implying that these summaries model the data well. 
Moreover, for highly structured datasets, such as \emph{DNA}, \emph{Kosarak}, and \emph{Plants}, this improvement is very large, and only a handful of itemsets are required to describe their main structure.

Comparing the two tables, we see that the number of itemsets selected by {\sc bic} compared to {\sc mdl} tends to be the same or a bit higher, indicating that {\sc bic} is more permissive in adding itemsets---an expected result. 
If we (crudely) interpret the raw {\sc bic} and {\sc mdl} scores as negative log-likelihoods, we see that {\sc bic} achieves lower, and hence better scores. This follows naturally from the larger collections of itemsets that {\sc bic} selects; the larger the collection, the more information it provides, and hence the higher the likelihood.

\begin{table}[t]
\caption{Statistics of the discovered summaries using \textsc{bic}, with respect to the background information $\ifam{B}$. Shown are the number of itemsets $k$,  
the wall clock time, the score of $\ifam{C}$, and the score of the empty collection. (Lower scores are better.)  Values for $k$ identified as optimal by \textsc{bic} are given in boldface. }

\label{tab:bicscores}
\centering
\begin{tabular*}{0.90\linewidth}{@{\extracolsep{\fill}}l rrrr}
\toprule
\textit{Dataset} & $k$ & \textit{time} & $\bicscore{\coll, \db ; \bg}$ & $\bicscore{\emptyset, \db ; \bg}$ \\
\midrule
Independent & \bf{7} &2m14s & 4\,494\,219 &4\,494\,242\\ 
Markov & \bf{62} & 15m44s & 4\,518\,101 & 4\,999\,963\\ 
Mosaic & \bf{16} & 6m06s & 807\,603 & 2\,167\,861\\[0.5em] 
Abstracts & \bf{220} & 27m21s & 233\,483 & 237\,771\\ 
Accidents & \bf{74} & 18m35s & 24\,592\,869 & 25\,979\,244\\ 
Chess (kr--k) & \bf{67} & 83m32s & 766\,235 & 786\,651\\
DNA Amplification & \bf{204} & 4m31s & 79\,164 & 183\,121\\ 
Kosarak & 261 & 120m36s & 66\,385\,663 & 70\,250\,533\\ 
Lotto & \bf{29} & 0m31s & 64\,970 &  65\,099\\
Mammals & \bf{76} & 25m23s& 99\,733 & 119\,461\\ 
MCADD & 80 & 121m02s & 2\,709\,240 & 2\,840\,837\\
Mushroom & \bf{80} & 28m11s & 358\,369 & 441\,130\\ 
Plants & \bf{94} & 66m56s & 732\,145 & 1\,271\,950\\
Retail & 62 & 121m46s & 8\,352\,161 & 8\,437\,118\\ 
\bottomrule
\end{tabular*}
\end{table}
\begin{table}[t]
\caption{Statistics of the discovered summaries using \textsc{mdl}, with respect to the background information $\ifam{B}$. Shown are the number of itemsets $k$, the wall clock time, the description length of $\ifam{C}$, and the description length of the empty collection.  (Lower scores are better).  Values for $k$ identified as optimal by \textsc{mdl} are given in boldface.}
\label{tab:mdlscores}
\centering
\begin{tabular*}{0.90\linewidth}{@{\extracolsep{\fill}}l rrrr}
\toprule
\textit{Dataset} & $k$ & \textit{time} & $\mdlscore{\coll, \db ; \bg}$ & $\mdlscore{\emptyset, \db ; \bg}$ \\
\midrule
Independent & \bf{0} & 1m57s & 4\,494\,243 & 4\,494\,243\\ 
Markov & \bf{62} & 16m16s & 4\,519\,723 & 4\,999\,964\\ 
Mosaic & \bf{15} & 7m23s & 808\,831 & 2\,167\,861\\[0.5em] 
Abstracts & \bf{29} & 1m26s & 236\,522 & 237\,772\\ 
Accidents & \bf{75} & 19m35s & 24\,488\,943 & 25\,979\,245\\ 
Chess (kr--k) & \bf{54} & 76m20s & 767\,756  & 786\,652\\
DNA Amplification & \bf{120} & 1m06s & 96\,967 & 183\,122\\ 
Kosarak & 262 & 120m53s & 66\,394\,995 & 70\,250\,534\\ 
Lotto & \bf{0} & 0m01s & 65\,100 & 65\,100\\
Mammals & \bf{53} & 2m17s & 102\,127 & 119\,461\\ 
MCADD & 82 & 123m34s & 2\,700\,924 & 2\,840\,838 \\
Mushroom & \bf{74} & 22m00s & 361\,891 & 441\,131\\
Plants & \bf{91} &63m9s & 734\,558 & 1\,271\,951\\
Retail & 57 & 122m44s & 8\,357\,129 & 8\,437\,119\\ 
\bottomrule
\end{tabular*}
\end{table}

\subsection{Summary Evaluation}
Next, we inspect the discovered data summaries in closer detail. 

For the \emph{Independent} dataset we see that using {\sc mdl} the returned summary is empty, i.e., no itemset can improve on the description length of the background model, which is the independence model. 
In general, {\sc mdl} does not aim to find any underlying `true' model, but simply the model that it considers best, given the available data. In this case, however, we see that the discovered model correctly corresponds to the process by which we generated the data.
Using {\sc bic}, on the other hand, the algorithm discovers 7 itemsets. These itemsets have an observed frequency in the data which is slightly higher or lower than their predicted frequencies under the independence assumption. Among the 7 itemsets, the highest absolute difference between those two frequencies is lower than 0.3\%. While these itemsets are each significant by themselves, after correcting their p-values to account for Type I error (see Section~\ref{sec:pvalues}), we find that none of them are statistically significant.

For the \emph{Markov} data, we see that all itemsets in the discovered summary are very small (about half of them of size 2, the rest mostly of size 3 or 4), and they all consist of consecutive items. Since the items in this dataset form a Markov chain, this is very much in accordance with the underlying generating model. In this case, the summaries for {\sc bic} and {\sc mdl} are identical. 
Knowing that the data is generated by a Markov chain, we can compute the {\sc mdl} score of the `best' model. Given that $\bg$ contains the singleton frequencies, the model that fully captures the Markov chain contains all 49 itemsets of size two containing consecutive items. The description length for this model is 4\,511\,624, which is close to what our algorithm discovered.

The third synthetic dataset, \emph{Mosaic}, contains 5 itemsets embedded in a 1\%-noisy database. The five first itemsets returned, both by {\sc bic} and {\sc mdl}, are exactly those itemsets. These sets are responsible for the better part of the decrease of the score, as can be seen in Figure~\ref{fig:mosaic}, which for {\sc mdl} depicts the description length as a function of the summary size. After these first five highly informative itemsets, a further ten additional itemsets are discovered that help explain the overlap between the itemsets---which, because we construct a probabilistic model using itemset frequencies cannot be inferred otherwise---and while these further itemsets do help in decreasing the score, their effect is much less strong than for the first sets. After discovering fifteen itemsets, the next best itemset does not decrease the log likelihood sufficiently to warrant the {\sc mdl} model complexity penalty, and the algorithm terminates. 

\begin{figure}[t!]
   \centering
	\includegraphics[scale=1.2]{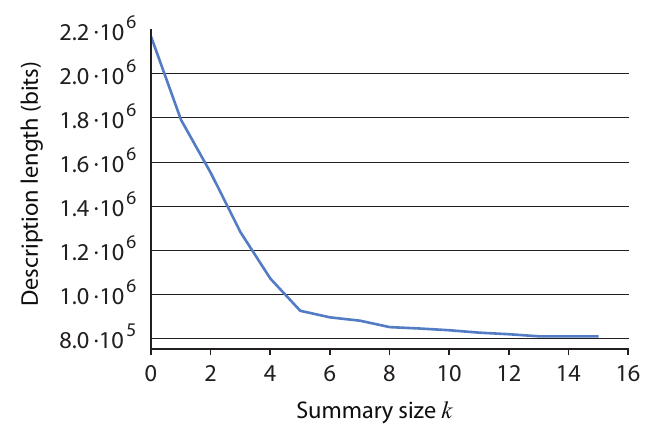}
   \caption{The description length of the \emph{Mosaic} dataset, as a function of the summary size $k$. The first five discovered itemsets correspond to the process which generated the dataset. The minimum {\sc mdl} score is attained at $k=15$.
}
   \label{fig:mosaic}
\end{figure}

For \emph{Lotto}, we see that using {\sc mdl}, our algorithm fails to discover any informative itemsets.\!\footnote{This is a desired result, assuming of course that the lottery is fair.} Using {\sc bic}, we find a summary of 29 itemsets, which are all combinations of 3 numbers, having been drawn between one and three times in the past twenty-odd years.  While for each itemset individually this is significantly lower than the expected absolute support (which is about 11), when we adjust for Type I error, they are no longer statistically significant. This gives evidence that in our setup, and with these data sizes, {\sc bic} may not be strict enough.
We additionally ran experiments using only row margins as background knowledge, i.e., using the fact that every transaction contains exactly seven items, the numbers of each lottery draw. Then, according to the maximum entropy background distribution, this implies that every number has exactly the same probability of being picked in a particular draw---namely \nicefrac{1}{6}. If our algorithm should find that a ball is drawn significantly more or less often, then it would be included in the model. Further, if there were any meaningful correlation between some of the numbers (positive or negative), these numbers would be included in a reported itemset. 
Using {\sc mdl}, our algorithm again finds no itemsets of interest. Using {\sc bic}, we find 7 itemsets (in a timespan of two hours, after which the algorithm was terminated), which are all combinations of five or six numbers, with absolute supports between 2 and 4, which is higher than expected. After applying Bonferroni adjustment, none of these itemsets turn out to be significant.

In the case of the \emph{DNA Amplification} data, our algorithm finds that the data can be described using 120 itemsets. As this dataset is banded, it contains a lot of structure \citep{garriga:11:banded}. Our method correctly discovers these bands, i.e., blocks of consecutive items corresponding to related genes, which lie on the same chromosomes. 
The first few dozen sets are large, and describe the general structure of the data. Then, as we continue, we start to encounter smaller itemsets, which describe more detailed nuances in the correlations between the genes. 
Figure~\ref{fig:dna-block-fig} depicts a detail of the \emph{DNA} dataset, together with a few of the itemsets from the discovered summary.
\begin{figure}[t!]
\centering
\includegraphics[scale=0.5]{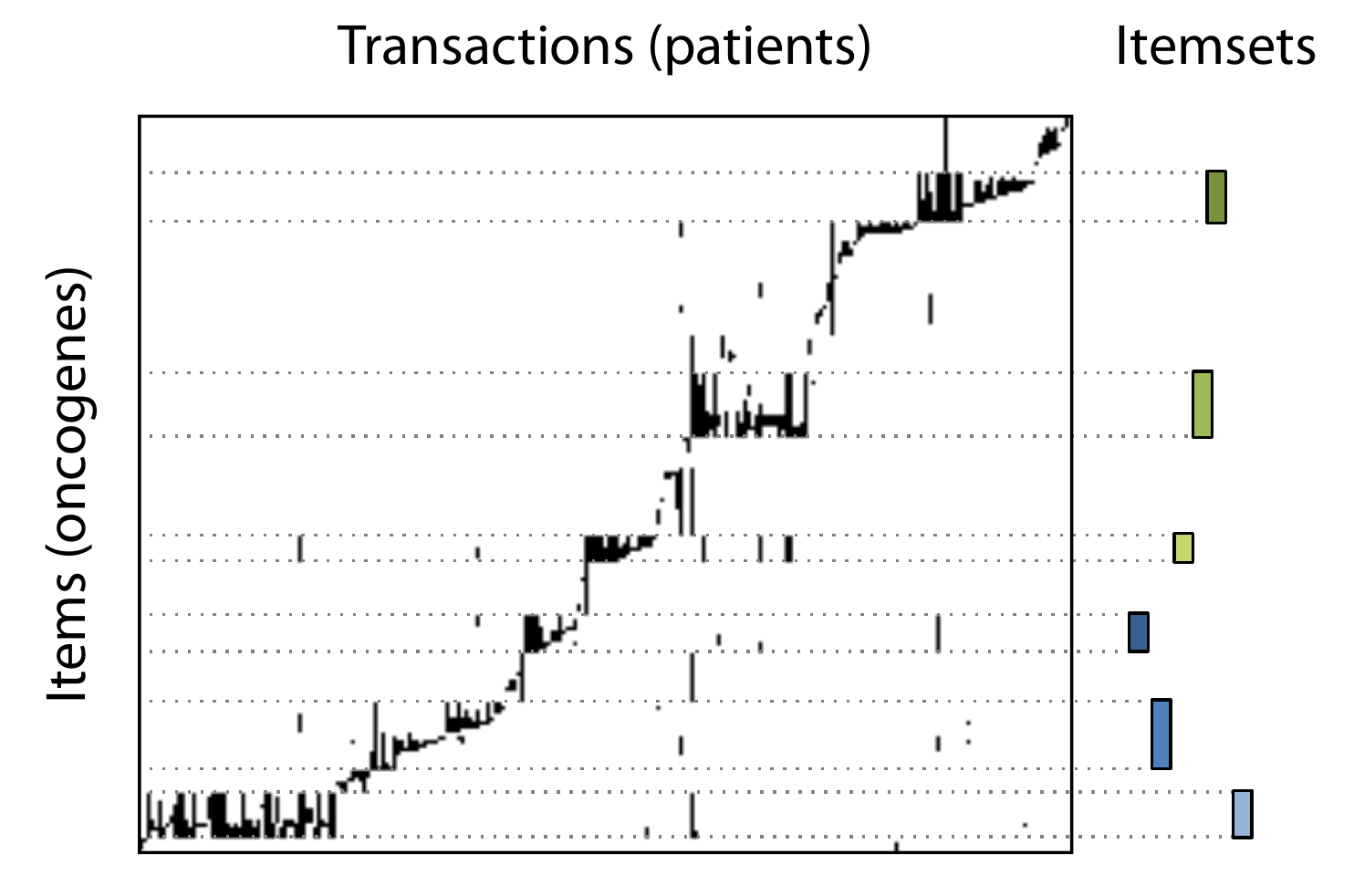}
\caption{Detail of the \emph{DNA Amplification} dataset (left), along with six of the discovered itemsets (right).}
\label{fig:dna-block-fig}
\end{figure}

\begin{figure}[t]
\centering
\hphantom{.}
\hfill
\includegraphics[scale=0.55]{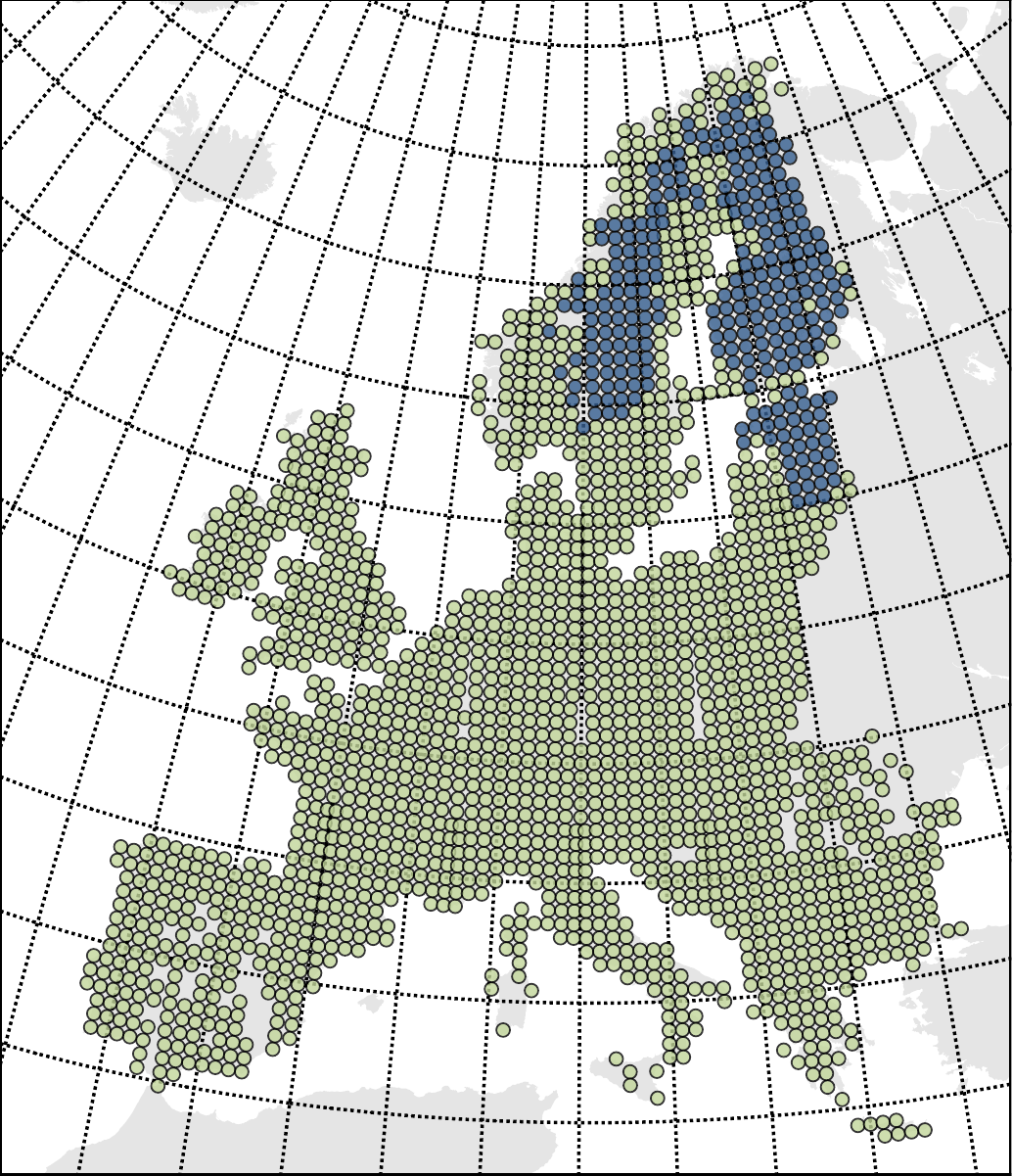}
\hfill
\includegraphics[scale=0.55]{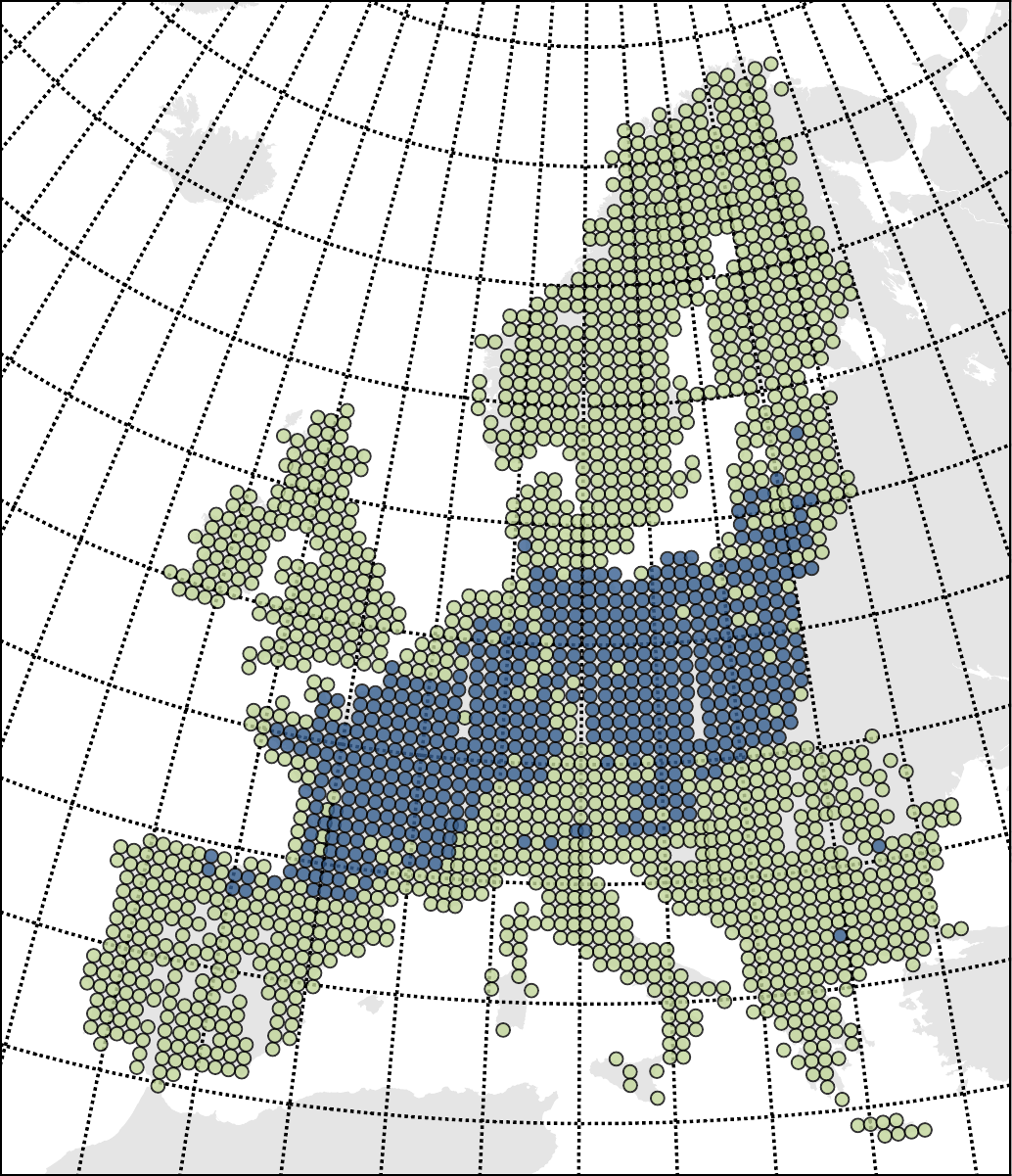}
\hfill
\hphantom{.}\\
\vspace{-0.5em}
\hphantom{.}
\hfill
\includegraphics[scale=0.55]{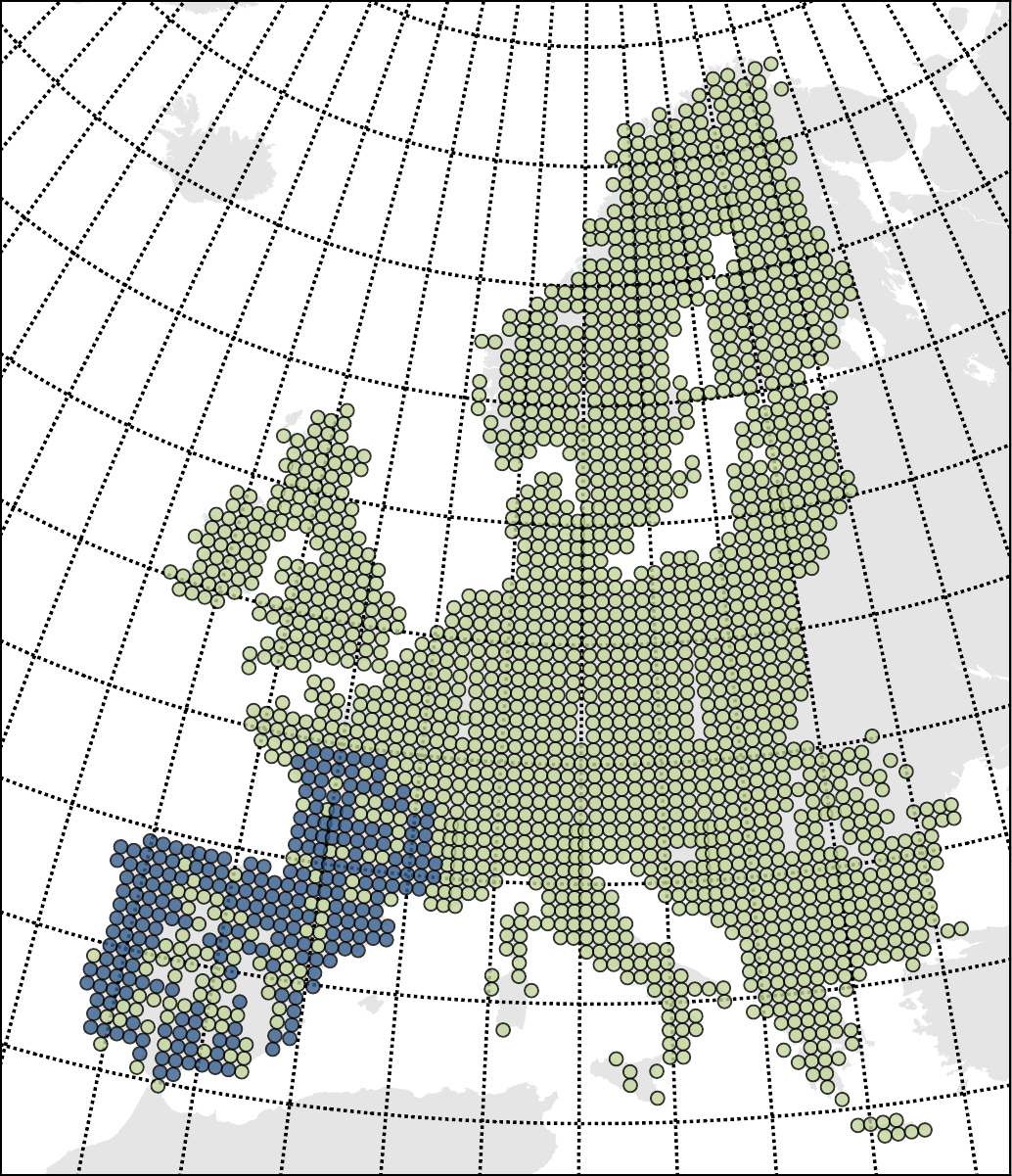}
\hfill
\includegraphics[scale=0.55]{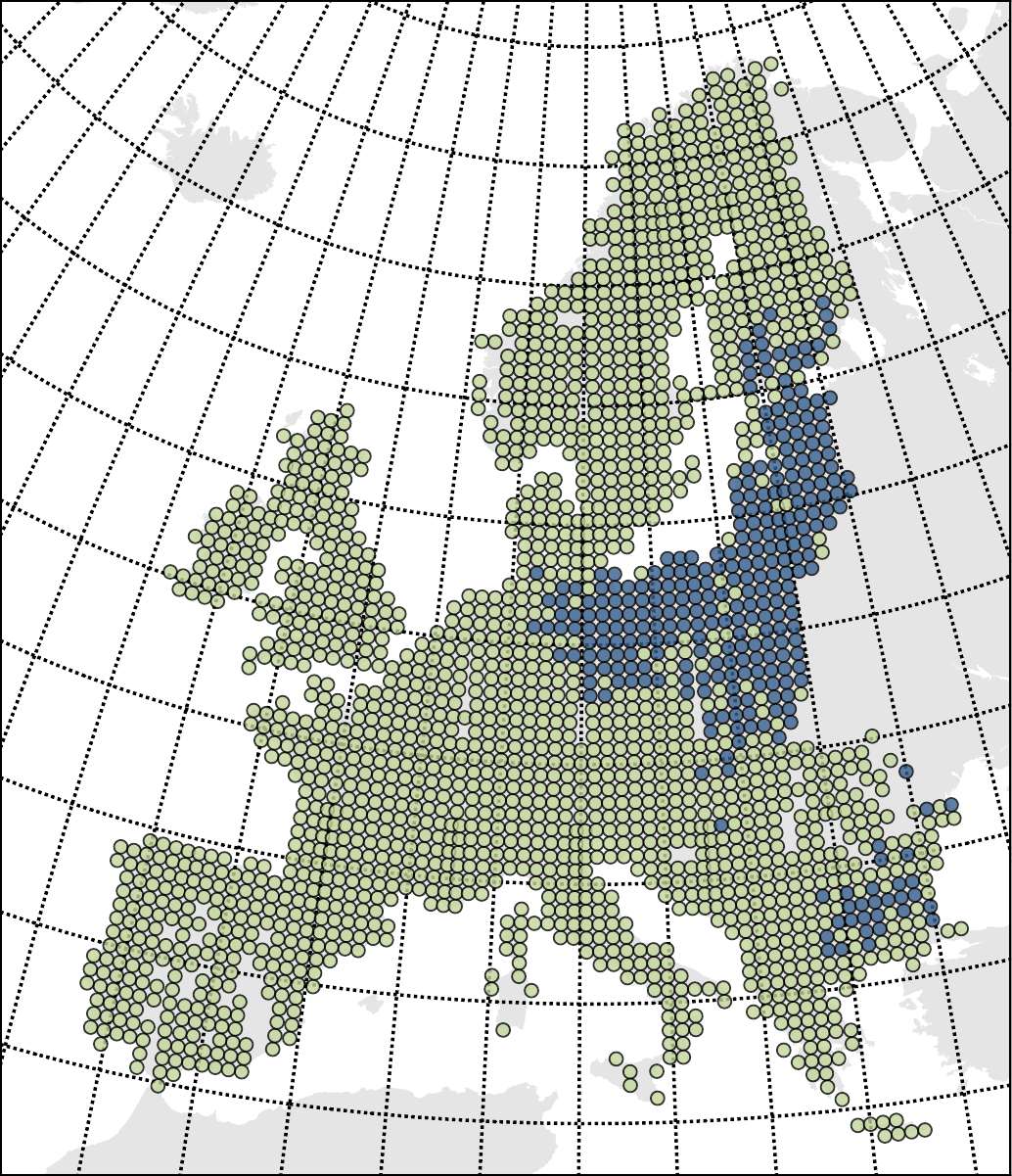}
\hfill
\hphantom{.}
\caption{For the \emph{Mammals} dataset, for four of the itemsets discovered by our method, we depict the transactions (i.e., locations) that support that itemset in blue (dark). A transaction supports the itemset if all of the mammals identified by the set have been recorded in the data to occur in that area.}
\label{fig:mammals}
\end{figure}

For the \emph{Chess (kr--k)} and \emph{Mushroom} datasets, we also ran experiments including the row margins, since we know that these datasets originated from categorical data, and hence the margins of these datasets are fixed. As remarked in Section~\ref{sec:background}, this increases the runtime of the algorithm by a factor $N=\abs{A}$. For both \textsc{bic} and \textsc{mdl}, the algorithm therefore took longer to execute, and was terminated after two hours for both datasets. The size of the summaries in all cases was equal to 6. Comparing between them reveals that some (but not all) of the itemsets also occur in the summaries that did not use the row margins as background knowledge, which leads us to conclude that the use of row margins as background information for these datasets does not have a substantial effect on the discovered summaries, at least not for the top-six.

The items in the \emph{Mammals} data are European mammals, and the transactions correspond to their occurrences in geographical areas of (50$\times$50)~{km}\textsuperscript{2}. The discovered itemsets represent sets of mammals that are known to co-exist in certain regions. For instance, one such set contains the Eurasian elk (moose), the mountain hare, the Eurasian lynx, and the brown bear, which are all animals living in colder, northern territories. Going back to the transactions of the data, we can also look at the areas where these itemsets occur. Figure~\ref{fig:mammals} depicts the geographical areas for some of the discovered itemsets. Some clear regions can be recognized, e.g., Scandinavia (for the aforementioned itemset), Central Europe, the Iberian peninsula, or Eastern Europe.

\citet{heikinheimo:07:mammals} considered the same data, annotated with environmental data, and applied mixture modelling and $k$-means clustering to identify regions with similar mammal presence.
When comparing to these expert-validated results, we observe a strong correlation to the regions (i.e., itemsets) MTV discovers. Although our method does not hard-cluster the map, we observe that our itemsets relate to clusters are various level of detail; the region identified in the top-right of Fig.~\ref{fig:mammals} shows strong correlation to the large `mainland' cluster of~\cite{heikinheimo:07:mammals} at $k=2,3$ whereas the mammals identified in the bottom-left plot of Fig.~\ref{fig:mammals} correspond to 2 clusters that only show up at $k\in[11,13]$.

The \emph{Plants} dataset is similar to the \emph{Mammals} data, except than now the plants form transactions, and the items represent the U.S.\ and Canadian states where they occur. The discovered itemsets, then, are states that exhibit similar vegetation. Naturally, these are typically bordering states. For instance, some of the first discovered itemsets are $\{\textsc{ct, il, in, ma, md, nj, ny, oh, pa, va}\}$ (North-Western states), $\{\textsc{al, ar, ga, ky, la, mo, ms, nc, sc, tn}\}$ (South-Western states), and $\{\textsc{co, id, mt, or, ut, wa, wy}\}$ (North-Eastern states).

Finally, for the \emph{MCADD} data, we find about 80 itemsets after running the algorithm for two hours. 
The attributes in this dataset are measured fatty acid concentrations, ratios between these, and some further biochemical parameters, which have all been discretized. The items therefore are attribute-value pairs. In the discovered summary, we see that the selected itemsets mostly consist of items corresponding to a few known key attributes. Furthermore, we can identify strongly correlated attributes by regarding those combinations of attributes within the itemsets selected in the summary.
As an example, one of the first itemsets of the summary corresponds to particular values for attributes $\{\mathit{MCADD}, C8, \frac{C8}{C2},\frac{C8}{C10},\frac{C8}{C12}\}$, respectively the class label, an acid, and some of its calculated ratios. This acid and its ratios, and the identified values, are commonly used diagnostic criteria for screening MCADD, and were also discovered in previous in-depth studies \citep{baumgartner:05:modelling,vandenbulcke:11:data}, and similarly for other itemsets in the summary.

\subsection{Comparison with Other Methods}

In Table~\ref{tab:abstracttop}, we give the top-10 itemsets in the {\em Abstracts} dataset, as discovered by our algorithm (using \textsc{mdl}), the method by \citet{konto:10:sdm}, the compression-based {\sc Krimp} algorithm~\citep{vreeken:11:krimp}, and Tiling~\citep{geerts:04:tiling}.
We see that our algorithm discovers important data mining topics such as {\em support vector machines}, {\em naive bayes}, and {\em frequent itemset mining}. 
Further, there is little overlap between the itemsets, and there is no variations-on-the-same-theme type of redundancy.

The results of the Information-Theoretic Noisy Tiles algorithm by \citet{konto:10:sdm}, based on the Information Ratio of tiles, are different from ours, but seem to be more or less similar in quality for this particular dataset.

The {\sc Krimp} algorithm does not provide its resulting itemsets in an order, so in order to compare between the different methods, following its {\sc mdl} approach, we selected the top-10 itemsets from the code table that have the highest usage, and hence, shortest associated code. From a compression point of view, the items in these sets co-occur often, and thus result in small codes for the itemsets. 
Arguably, this does not necessarily make them the most interesting, however, and we observe that some rather general terms such as {\em state [of the] art} or {\em consider problem} are ranked highly.

\begin{table}[h]
\caption{The top-10 itemsets of the \emph{Abstracts} dataset for our \textsc{mtv} algorithm using \textsc{mdl} (top left), \citet{konto:10:sdm} (top right), \textsc{Krimp}~\citep{vreeken:11:krimp} (bottom left), and Tiling~\citep{geerts:04:tiling} (bottom right).}
\label{tab:abstracttop}
\small
\centering
\begin{tabular*}{0.99\linewidth}{@{\extracolsep{\fill}}c c c}
\cmidrule{1-1}\cmidrule{3-3}
\multicolumn{1}{c}{\textsc{mtv}} && \multicolumn{1}{c}{\textit{Information-Theoretic Noisy Tiles}} \\
\cmidrule{1-1}\cmidrule{3-3}
support vector machin svm && support vector machin \\
associ rule mine && effici discov frequent pattern mine algorithm \\
nearest neighbor && associ rule mine algorithm database \\
frequent itemset mine && train learn classifi perform set \\
naiv bay && frequent itemset\\
linear discrimin analysi lda && mine high dimensional cluster \\
cluster high dimension && synthetic real \\
state art && time seri \\
frequent pattern mine algorithm && decis tree classifi \\
synthet real && problem propos approach experiment result \\[0.5em]
\cmidrule{1-1}\cmidrule{3-3}
\multicolumn{1}{c}{\textsc{Krimp}} && \multicolumn{1}{c}{\textit{Tiling}} \\
\cmidrule{1-1}\cmidrule{3-3}
\hphantom{experiment propos problem approach result}\vspace{-0.9em}&&
\hphantom{experiment propos problem approach result}\\
algorithm experiment result set && algorithm mine \\
demonstr space  && algorithm base \\
larg databas  && result set\\
consid problem  && approach problem\\
knowledg discoveri && propos method \\
experiment demonstr && experiment result\\
rule mine associ databas  && algorithm perform\\
algorithm base approach cluster  && model base\\
state art  && set method\\
global local  && algorithm gener\\
\end{tabular*}
\end{table}

Finally, for Tiling we provide the top-10 tiles of at least two items, i.e., the ten tiles whose support times their size is highest. Without the minimum size constraint, only singleton itemsets are returned, which, although objectively covering the largest area individually, are not very informative. Still, the largest discovered tiles are of size two, and contain quite some redundancy, for instance, the top-10 contains only 13 (out of 20) distinct items.

Each of these methods formalize the task of summarization differently, and  hence optimize objective scores quite different from ours. It therefore is impossible to fairly and straightforwardly quantitatively compare to these methods; neither construct a itemset-frequency-prediction model as we do, and not even result itemsets together with their frequencies. Tiling and Information Theoretic Tiling both result in tiles, sets of itemsets with tid-lists, while \textsc{Krimp} code tables contain itemsets linked to relative \emph{usage} frequencies of covering the data without overlap. It hence makes as little sense to score \textsc{Krimp} itemsets by our score, as it would to consider our itemsets as a code table.

\subsection{Significant Itemsets}
\label{sec:pvalues}
Next, we investigate the significance of the itemsets that are included in the summaries we discover, as well as the significance of the itemsets in $\ifam{F}$ that were not included. 
To properly compute the p-values of the itemsets, we employ holdout evaluation, as described by \citet{webb:07:discovering}. That is, we equally split each dataset into an exploratory (training) set $\db_{e}$ and a holdout (test) set $\db_{h}$, apply our algorithm to the exploratory data, and evaluate the itemsets on the holdout set. Since $\db_{e}$ contains less data than $\db$, the discovered models tend to contain fewer itemsets (i.e., as there is less data to fit the model on, \textsc{bic} and \textsc{mdl} both allow for less complex models).

The significance of an itemset $X$ is evaluated as follows. We compute its estimated probability $\pemp(X=1)$ (using either $\bg$, or $\bg$ and $\coll$, consistent with $D_{e}$). This is the null hypothesis. Then, we calculate the two-tailed p-value given the observed frequency in the holdout data, $\freq{X}$. Let us write $d=\abs{\db_{h}}=\abs{\db}/2$, $f=d\cdot\freq{X}$, and $p=\pemp(X=1)$. The p-value of $X$ expresses the probability of observing an empirical frequency $q_{\db_{h}}(X)$ at least as extreme (i.e., improbable) as the observed frequency $\freq{X}$, with respect to the model, according to a binomial distribution parametrized by $d$ and $p$
\[
	B(d;p)(f)={d \choose f} p^{f} (1-p)^{(d-f)}\;.
\]
Assuming that $\freq{X} \geq \pemp(X=1)$, we calculate the two-tailed p-value of $X$ as
\begin{align*}
	\mathrm{p\mhyphen value} & = \mathrm{Prob}(q_{\db_{h}}(X)\geq f \mid p) + \mathrm{Prob}(q_{\db_{h}}(X) \leq f' \mid p)\\
	&= \sum_{i=f}^{d} {d \choose i} p^{i} (1-p)^{d-i} + \sum_{i=0}^{f'} {d \choose i} p^{i} (1-p)^{d-i}
\end{align*}
where \[f' = \max \{f' \in [0,dp) \mid B(d;p)(f') \leq B(d;p)(f) \}\;.\]
The p-value for the case $\freq{X} < \pemp(X=1)$ is defined similarly.

It is expected that the itemsets that our algorithm discovers are significant with respect to the background model. Simply put, if the frequency of an itemset is close to its expected value, it will have a high p-value, and thus it will not be significant. Moreover, it will also have a low heuristic value $h(X)$, and hence it will not greatly improve the model. The following demonstrates this connection between $\mathit{kl}$ (and hence $h$), and the p-value. Using Stirling's approximation, we can write the logarithm of the binomial probability $B(d;p)(f)$ as follows.
\begin{align*}
	\log {d \choose f} p^{f} (1-p)^{d-f} &\approx d\log d - f\log f - (d-f) \log (d-f) \\
	& \quad+ f\log p + (d-f) \log (1-p)\\
	&= - f \log \frac{f}{dp} - (d-f) \log \frac{d-f}{d(1-p)}\\
	&= -d \cdot \mathit{kl}(X)
\end{align*}
Therefore, when we maximize the heuristic $h$, we also indirectly minimize the p-value. Note, however, that in our algorithm we do not employ a significance level as a parameter; we simply let \textsc{bic} or \textsc{mdl} decide when to stop adding itemsets. Further, we can also take the complexity of the itemsets into account.

In Table~\ref{table:pvalues} we show the number of significant selected and unselected itemsets. Since we are testing multiple hypotheses, we apply  Bonferroni adjustment, to avoid Type I error---falsely rejecting a true hypothesis \cite{shaffer:95:multiple}. That is, if we were to test, say, 100 true hypotheses at significance level 0.05, then by chance we expect to falsely reject five of them. Therefore, at significance level $\alpha$, each p-value is compared against the adjusted threshold $\alpha/n$ where $n$ is the number of tests performed.

The first column of Table~\ref{table:pvalues} shows the number of itemsets in $\coll$ that are significant with respect to the background knowledge $\bg$. In general, we see that all itemsets are significant. However, for the \emph{Accidents} dataset, e.g., we find two itemsets that are not significant with respect to $\bg$. Nevertheless, they are not redundant in this case; upon inspection, it turns out that these itemsets (say, $X_{1}$ and $X_{2}$) are subsets of another itemset (say, $Y$) in $\coll$ that was significant. After adding $Y$ to the model, the estimates of $X_{1}$ and $X_{2}$ change, causing them to become significant with respect to the intermediate model. A similar observation is made for the \emph{Chess (kr--k)} and \emph{Mushroom} datasets. In the second column, we therefore also show the number of itemsets $X_{i+1} \in \coll$ that are significant with respect to the previous model $\coll_{i}$. In this case we see that indeed all itemsets are significant, from which we conclude that all itemsets really contribute to the summary, and are not redundant. 

Next, we compute the p-values of the itemsets in the candidate set $\ifam{F}$ that were not included in $\coll$. Since for all datasets the candidate set $\ifam{F}$ is very large, we uniformly sample 1\,000 itemsets from $\ifam{F} \setminus \coll$, and compute their p-values with respect to $\pemp_{\bg,\coll}$. We see that for seven of the datasets, there are few significant itemsets, which means that $\coll$ captures the data quite well. For two datasets (\emph{Mushroom} and \emph{Retail}), a few hundred are significant, but still many are not. For the five remaining datasets, however, we see that almost all itemsets are significant. Nonetheless, this does not automatically mean that the discovered models are poor. That is, apart from considering the deviation of an itemset's frequency, we also consider its complexity and the complexity of a model as a whole. This means that even if the observed frequency of an itemset is surprising to some degree, it may be too complex to include it; upon inspection, we indeed find that among the sampled itemsets, there tend to be many large ones.

\begin{table}[ht]
\caption{The number of significant itemsets in the discovered summary $\coll$, with respect to background knowledge $\bg$ and each intermediate model $\coll_{i}$, and the number of significant itemsets among 1\,000 itemsets sampled from $\ifam{F} \setminus \coll$, denoted by $\ifam{S}$. We use a significance level of 0.05, and Bonferroni-adjusted p-values.}
\label{table:pvalues}
\centering
\begin{tabular}{lrrrcrr}
\toprule
& \multicolumn{3}{c}{$X \in \coll$} && \multicolumn{2}{c}{$X \in \ifam{S}$} \\
 \cmidrule{2-4} \cmidrule{6-7}
 & \parbox{4em}{\textit{\raggedleft\# signif.\\\hfill w.r.t.\ $\bg$}} & \parbox{4em}{\textit{\raggedleft\# signif.\\\hfill w.r.t.\ $\coll_{i}$}} & $\abs{\coll}$ && \parbox{4em}{\textit{\raggedleft\# signif.\\\hfill w.r.t.\ $\coll$}} & $|\ifam{S}|$\\
\midrule
Independent & 0 & 0 & \bf{0} && 0 & 1\,000\\
Markov & 64 & 64 & \bf{64} &&  6 & 1\,000\\
Mosaic & 14 & 14 & \bf{14} && 0 & 1\,000\\[0.5em]
Abstracts & 12 & 12 & \bf{12} && 14 & 1\,000\\
Accidents  & 69 &71 & \bf{71} && 883 & 1\,000\\
Chess (kr--k) & 42 & 42 & \bf{43} && 37 & 1\,000\\
DNA Amplification & 87 & 87 & \bf{87} && 990 & 1\,000\\
Kosarak & 268 & 268 & 268 && 993 & 1\,000\\
Lotto & 0 & 0 & \bf{0} && 0 & 1\,000\\
Mammals & 39 & 39 & \bf{39} && 986 & 1\,000\\
MCADD & 61 & 61 & 61 && 89 & 1\,000\\
Mushroom & 59 & 63 & \bf{63} && 139 & 1\,000\\
Plants & 87 & 87 & \bf{87} && 998 & 1\,000\\
Retail & 65 & 65 & 65 && 302 & 1\,000\\
\bottomrule
\end{tabular}
\end{table}

%% file: discussion.tex
\section{Discussion}
\label{sec:discussion}

The approach introduced in this paper fulfills several intuitive expectations one might have about summarization, such as succinctness, providing a characteristic description of the data, and having little redundancy. 
The experiments show that quantitatively we can achieve good {\sc bic} and {\sc mdl} scores with only a handful of itemsets, 
and that these results are highly qualitative and meaningful; moreover, we can discover them in a relatively short amount of time. In practice, we see that the results using \textsc{mdl} are slightly better than those using \textsc{bic}; the former tends to be a bit more conservative, in the sense that it does not discover spurious itemsets. Furthermore, using our \textsc{mdl} score, we have more control over the complexity of the summary, since it takes into account not only the summary size but also the sizes of the itemsets within it. 

\textsc{mdl} does not provide a free lunch. First of all, although highly desirable, it is not trivial to bound the score. For Kolmogorov complexity, which \textsc{mdl} approximates, we know this is incomputable. For our models, however, we have no proof one way or another. Furthermore, although \textsc{mdl} gives a principled way to construct an encoding, this involves many choices that determine what structure is rewarded. As such, we do not claim our encoding is suited for all goals, nor that it cannot be improved.
Since we take a greedy approach, and in each step optimize a heuristic function, we also do not necessarily find the globally optimal solution with respect to \textsc{mdl} or \textsc{bic}. However,
we argue that in practice it is not strictly necessary to find an optimal summary, but to find one that provides a significant amount of novel information. 

In this paper we consider data mining as an iterative process. By starting off with what we already know---our background knowledge---we can identify those patterns that are the surprising to us. Simply finding the itemset that is most surprising, is a problem that \citet{hanhijarvi:09:tell} describe as `{tell me something I don't know}'\!. When we repeat this process, in the end, we will have identified a group of itemsets that `{tell me all there is to know}' about the data. Clearly, this group strongly overfits the data. This is where the {\sc mdl} principle provides a solution, as it automatically identifies the most informative group. Hence, we paraphrase our approach as `{tell me what I need to know}'\!. As such, by our Information Theoretic approach it can be seen as an instantiation of the general iterative framework recently proposed by~\citet{debie:11:inftheoryframework}.

The view that we take here on succinctness and non-redundancy is fairly strict. 
Arguably, there are settings conceivable where limited redundancy (at the cost of brevity) can give some robustness to a technique, or provide alternative insights by restating facts differently. However, this is not the intention of this paper, and we furthermore argue that our method can perfectly be complemented by techniques such as redescription mining \citep{zaki:05:reasoning}.

Data mining is not only an iterative process, but also an interactive one. The {\sc mtv} algorithm above simply returns a set of patterns to the user, with respect to his or her background knowledge. However, in practice we might want to dynamically guide the exploration and summarization process. That is, we may want to add or remove certain itemsets, next let the algorithm add some more itemsets, etc. Our algorithm can easily be embedded into an interactive environment. Deleting an itemset from a summary is quite straightforward; we just collapse the transaction partition, and re-run the iterative scaling algorithm to discover the parameters of the simplified model.

Even though in this paper we significantly extend upon \citet{mampaey:11:tell}, there are some further improvements possible for our method. For instance, one problem setting in which our method is applicable, is that of finding the best specialization of a given itemset $X$. That is, to identify the superset $Y$ of $X$ that provides the best score $s(\coll \cup \{Y\})$. This setup allows experts to interactively discover interesting itemsets. As part of future work, we are currently investigating this in practice for finding patterns in proteomics and mass-spectrometry data.

The experiments showed that we discover high-quality summaries, and that we can efficiently compute the maximum entropy model for a given collection of itemsets in practice---even though the latter is an NP-hard problem in general. In specific cases, however, there is room for further optimization. For computational reasons, we split up the distribution into smaller groups, by restricting the number of items per group. However, this does imply that exactly modeling a (very) long Markov chain, for instance, is not possible, as only parts up to $n$ items will be modeled. While a complete Markov chain can easily be described by our model, computing it may prove to be difficult in practice. It would be interesting to see then, how modeling techniques as decomposition, for instance, using junction trees \cite{cowell:99:probabilistic}, could be combined with our methods.

We see several further possibilities for improving upon our current, unoptimized, implementation of \textsc{mtv}. For example, massive parallelization can be applied to the search for the most informative itemset. This requires computing itemset probabilities with respect to a single model, for many candidates, which can easily be done in parallel. This can also benefit pruning in this phase, since we can use the maximum heuristic value over all processes. Another option is to simplify the computation of the row margin distribution, which adds a factor $N$ to the runtime of the algorithm. Rather than computing the probability that a transaction contains a certain number of items, we could group these probabilities in to a coarser granularity, to reduce this factor.

%% file: conclusion.tex
\section{Conclusion}
\label{sec:conclusion}

We introduced a well-founded method for iteratively mining non-redundant collections of itemsets that form high-quality summaries of transaction data.
By employing the Maximum Entropy principle, we obtain unbiased probabilistic models of the data, through which we can identify informative itemsets, and subsequently iteratively build succinct summaries of the data by updating our model accordingly. As such, unlike static interestingness models, our approach does not return patterns that are redundant with regard to what we have learned, or already knew as background knowledge, and hence the result is kept succinct and maximally informative. 

To this end, we presented the {\sc mtv} algorithm for mining informative summaries, which can either mine the top-\emph{k} most informative itemsets, or by employing either the Bayesian Information Criterion ({\sc bic}) or the Minimum Description Length ({\sc mdl}) principle, we can automatically identify the set of itemsets that as a whole provides a high quality summary of the data. Hence, informally said, our method `tells you what you need to know' about the data.

Although in the general case modeling by Maximum Entropy is NP-hard, we showed that in our case we can do so efficiently using Quick Inclusion-Exclusion. Furthermore, {\sc mtv} is a one-phase algorithm; rather than picking itemsets from a pre-materialized user-provided candidate set, it calculates on-the-fly the supports for only those itemsets that stand a chance of being included in the summary.

Experiments show that we discover succinct summaries, which correctly identify important patterns in the data. The resulting models attain high log-likelihoods using only few itemsets, are easy to interpret, contain significant itemsets with low p-values, and for a wide range of datasets are shown to give highly intuitive results.

%% file: appendix.tex
\appendix
\section*{APPENDIX}
\section{Proof of Corollary~\lowercase{\ref{cor:likelihood}}}
\begin{apxcorollary43}[]
The log-likelihood of the maximum entropy distribution $\pemp_{\langle \coll, \faai\rangle}$ for a collection of itemsets and frequencies $\langle \coll, \faai\rangle$ is equal to
\begin{align*}
	\log \pemp_{\langle \coll, \faai\rangle}\fpr{\db} &= \abs{\db}\bigg(\log u_0 + \sum_{(X_{i}, f_{i}) \in \langle \coll, \faai\rangle} f_{i} \log u_{X_{i}}\bigg)\,.
\end{align*}
\end{apxcorollary43}
\begin{proof} From Theorem~\ref{thr:exponential}, we have that $\pemp_{\langle \coll, \faai\rangle}\fpr{A = t} = u_0 \prod_{X \in \coll} u_X^{S_X(t)}$. Hence
\begin{align*}
	\log \pemp_{\langle \coll, \faai\rangle}\fpr{\db}  &= \sum_{t\in\db} \log \pemp_{\langle \coll, \faai\rangle} (A=t)\\
	&= \sum_{t\in\db} \left(\log u_0 + \sum_{X_i \in \coll} S_{X_i}(t) \log u_{X_i} \right)\\
	&= \abs{\db}\bigg(\log u_0 + \sum_{(X_{i}, f_{i}) \in \langle \coll, \faai\rangle} f_{i} \log u_{X_{i}}\bigg) \\
	&= -\abs{\db}\ent{\pemp_{\langle \coll, \faai\rangle}}\,.
\end{align*}
The transition from the second to the third line follows from the fact that $\sum_{t\in\db} S_{X_i}(t)=\supp{X}$.
\end{proof}

\section{Proof of Theorem~\lowercase{\ref{theorem:4.5}}}
\begin{apxtheorem45}[]
Assume a collection of itemsets $\coll = \enset{X_1}{X_k}$ and let $\pemp_\coll$ be the maximum entropy model,
computed from a given dataset $\db$.
Assume also an alternative model $r(A = t \mid f_1, \ldots, f_k)$, where $f_i = \freq{X_i \mid D}$, 
that is, a statistical model parametrized by the frequencies of $\coll$. Assume
that for any two datasets $\db_1$ and $\db_2$, where $\freq{X \mid \db_1}
= \freq{X \mid \db_2}$ for any $X \in \coll$, it holds that
\[
	1/\abs{D_1} \log r(D_1 \mid f_1, \ldots, f_k) = 1/\abs{D_2} \log r(D_2 \mid f_1, \ldots, f_k)\;.
\]
Then $\pemp_{\coll}(\db) \geq r(\db)$ for any dataset $\db$.
\end{apxtheorem45}

\begin{proof}
There exists a sequence of finite datasets $D_j$ such that $\freq{X_i \mid D_j} = f_i$
and $q_{D_j} \to \pemp_\coll$. To see this, first note that $f_i$ are rational numbers
so that the set of distributions $\ifam{P}_\coll$ is a polytope with faces defined
by rational equations. In other words, there is a sequence of distributions $p_j$ with only rational
entries reaching $\pemp$. Since a distribution with rational entries can be represented
by a finite dataset, we can have a sequence $D_j$ such that $q_{D_j} = p_j$.

Now as $j$ goes to infinity, we have
\[
	\frac{1}{\abs{D}}\log r(D) = \frac{1}{\abs{D_j}} \log r(D_j) = \sum_{t \in \trans} q_{D_j}(A = t) \log r(A = t) \to \sum_{t \in \trans} \pemp_\coll(A = t) \log r(A = t)\;.
\]
Since the left side does not depend on $j$ we actually have a constant sequence.
Hence, 
\[
	0 \leq \kl{\pemp_\coll}{r} = \ent{\pemp_\coll} - \sum_{t \in \trans} \pemp_\coll(A = t) \log r(A = t) = 1/\abs{D}(\log \pemp_\coll(D) - \log r(D))\;,
\]
which proves the theorem.
\end{proof}

\section{Proof of Theorem~\lowercase{\ref{theorem:itscorrect}}}

First, we state two lemmas, which are then used in the proof of the subsequent theorem.

\begin{lemma}
\label{lem:closure2}
Let $\ifam{G}$ and $\ifam{H}$ be itemset collections such that $\ifam{G} \subseteq \ifam{H} \subseteq \closure{\ifam{G}}$, then
$\closure{\ifam{H}} = \closure{\ifam{G}}$.
\end{lemma}

\begin{proof}
Write $\ifam{F} = \closure{\ifam{G}}$ and let
$U = \bigcup_{X \in \ifam{G}} X$, $V = \bigcup_{X \in \ifam{H}} X$, and $W = \bigcup_{X \in \ifam{\ifam{F}}} X$.
By definition we have $U \subseteq V$ which implies that $\closure{\ifam{G}} \subseteq \closure{\ifam{H}}$. Also since $V \subseteq W$,
we have $\closure{\ifam{H}} \subseteq \closure{\ifam{F}} = \closure{\ifam{G}}$, where the second equality follows
from the idempotency of closure.
\end{proof}

\begin{lemma}
\label{lem:closure3}
Let $\ifam{G}$ be an itemset collection and let $Y \notin \ifam{G}$ be an itemset. Assume that there is a $t \in \trans$
such that $S_X(t)$ for every $X \in \ifam{G}$ and $S_Y(t) = 0$. Then $Y \notin \closure{\ifam{G}}$.
\end{lemma}

\begin{proof}
$S_Y(t) = 0$ implies that there is $a_i \in Y$ such that $t_i = 0$. Note that $a_i \notin X$ for any $X \in \ifam{G}$, otherwise
$S_X(t) = 0$. This implies that $Y \nsubseteq \bigcup_{X \in \ifam{G}} X$ which proves the lemma.
\end{proof}

With the lemmas above, we can now prove the main theorem.

\begin{apxtheorem418}[]
Given a collection of itemsets $\coll=\{X_{1},\ldots,X_{k}\}$, let $\trans_{\coll}$ be the corresponding partition with respect to $\coll$.
The algorithm {\sc QieBlockSizes} correctly computes the block sizes $e(T)$ for $T\in \trans_{\coll}$.
\end{apxtheorem418}
\begin{proof}
Let us denote $e^{0}(T)=c(T)$ for the initialized value, and let $e^{i}(T)$ be
the value of $e(T)$ after the execution of the $i$th iteration of {\sc
QieBlockSizes} for $i=1,\ldots,k$.

Let us write
\[
	S^{i}(\ifam{G}) = \set{t \in \trans \mid S_{X}(t)=1 \mathrm{\ for\ } X \in \ifam{G} \mathrm{\ and\ } S_{X}(t)=0 \mathrm{\ for\ } X \in \coll_{i} \setminus \ifam{G}}\;.
\]

The following properties hold for $S^i$: 
\begin{enumerate}
\item \label{property:subset}
$S^i(\ifam{G}) \subseteq S^{i - 1}(\ifam{G})$ for any $\ifam{G}$ and $i > 0$.
\item \label{property:split}
If $X_i \notin \ifam{G}$, then $S^{i - 1}(\ifam{G}) = S^i(\ifam{G}) \cup S^{i - 1}(\ifam{G} \cup \set{X_i})$,
and $S^{i - 1}(\ifam{G} \cup \set{X_i}) \cap S^i(\ifam{G}) = \emptyset$.
\item \label{property:closure}
$S^i(\ifam{G}) = S^i(\closure{\ifam{G}, i})$ for any $\ifam{G}$.
\item \label{property:shift}
If $X_i \in \ifam{G}$, then $S^i(\ifam{G}) = S^{i - 1}(\ifam{G})$.
\end{enumerate}

Note that $\abs{S^k(\sets{T; \coll})} = e(T)$, hence to prove the theorem we will show by
induction that $e^i(T) = \abs{S^i(\sets{T; \coll})}$ for $i = 0, \ldots, k$ and for $T \in \trans_{\coll}$.

For $i=0$ the statement clearly holds. Let $i > 0$ and make an induction
assumption that $e^{i - 1}(T) = \abs{S^{i - 1}(\sets{T; \coll})}$. 
If $X_{i} \in \ifam{G}$, then by definition $e^i(T) = e^{i - 1}(T)$.
Property~\ref{property:shift} now implies that $S^i(\ifam{G}) = S^{i - 1}(\ifam{G})$, proving the induction step.

Assume that $X_{i} \notin \ifam{G}$. Let us define $\ifam{F} = \ifam{G} \cup \set{X_i}$
and let $\ifam{G}' = \closure{\ifam{F}, i - 1}$ and $\ifam{H} = \closure{\ifam{G}'}$.
Since it holds that $\ifam{F} \subseteq \ifam{G}' \subseteq \closure{\ifam{F}}$, Lemma~\ref{lem:closure2}
implies that $\ifam{H} = \closure{\ifam{F}}$.

Assume that $T'=\block{\ifam{G}'}=\emptyset$.
Then we have $e^i(T) = e^{i - 1}(T)$, hence we need to show
that $S^i(\ifam{G}) = S^{i - 1}(\ifam{G})$.  Assume otherwise.
We will now show that $\ifam{G}' = \ifam{H}$ and apply Lemma~\ref{lem:closure1} to conclude
that $\block{\ifam{G}'}\neq\emptyset$, which contradicts our assumption.

First note that if there would exist an $X \in \ifam{H} \setminus \ifam{G}'$, then $X \in \coll_{i - 1}$, since
\begin{align*}
	\ifam{H} \setminus \ifam{G}' &= \ifam{H} \setminus (\ifam{F} \cup (\closure{\ifam{F}} \setminus \coll_{i - 1})) \\
	&= \ifam{H} \setminus (\ifam{F} \cup (\ifam{H} \setminus \coll_{i - 1})) \\
	&\subseteq \ifam{H} \setminus (\ifam{H} \setminus \coll_{i - 1}) \\
	& \subseteq \coll_{i - 1}\;.
\end{align*}

Since $S^i(\ifam{G}) \subseteq S^{i - 1}(\ifam{G})$, we can choose by our
assumption $t \in S^{i - 1}(\ifam{G}) \setminus S^i(\ifam{G})$.  Since $t \in
S^{i - 1}(\ifam{G})$, we have $S_X(t) = 1$ for every $X \in \ifam{G}$. Also
$S_{X_i}(t) = 1$, otherwise $t \in S^i(\ifam{G})$.  Lemma~\ref{lem:closure3}
now implies that for every $X \in \coll_{i - 1} \setminus \ifam{F}$, it holds
that $X \notin \closure{\ifam{F}} = \ifam{H}$.  This proves that $\ifam{H} =
\ifam{G}'$, and Lemma~\ref{lem:closure1} now provides the needed contradiction. Consequently, $S^{i -
1}(\ifam{G}) = S^i(\ifam{G})$.

Assume now that $T' = \block{\ifam{G}'} \neq \emptyset$, i.e., there exists a $T' \in \trans_{\coll}$ such that $\ifam{G}' = \sets{T', \coll}$. The induction assumption now guarantees that $e^{i - 1}(\ifam{G}') = \abs{S^{i - 1}(\ifam{G}')}$.

We have
\begin{align*}
	\abs{S^i(\ifam{G})} & = \abs{S^{i - 1}(\ifam{G})} - \abs{S^{i - 1}(\ifam{F})} \tag*{Property~\ref{property:split}}\\
	                    & = \abs{S^{i - 1}(\ifam{G})} - \abs{S^{i - 1}(\ifam{G}')} \tag*{Property~\ref{property:closure}}\\
	                    & = e^{i - 1}(\ifam{G}) - e^{i - 1}(\ifam{G}')\;, \tag*{induction assumption}\\
\end{align*}
which proves the theorem.
\end{proof}